\documentclass[conference]{IEEEtran}
\IEEEoverridecommandlockouts

\usepackage{cite}
\usepackage{amsmath,amssymb,amsfonts,amsthm}
\usepackage{algorithmic}
\usepackage{graphicx}
\usepackage{textcomp}
\usepackage{enumitem}
\usepackage{multirow}
\usepackage{listings}
\usepackage{thmtools}
\usepackage{thm-restate}
\usepackage{hyperref}
\usepackage{cleveref}
\usepackage{pifont}
\usepackage{subfig}
\usepackage{booktabs}
\usepackage{array}
\usepackage{xspace}

\usepackage{listings, xcolor}
\newtheorem{lemma}{Lemma}
\newcommand{\xmark}{\ding{55}}
\newcommand{\cmark}{\ding{52}}
\newcommand{\mypara}[1]{\noindent{\bf {#1}.}\xspace}

\definecolor{verylightgray}{rgb}{.97,.97,.97}
\lstdefinelanguage{Solidity}{
    keywords=[1]{
        anonymous, assembly, assert, balance, break, call, callcode, case, catch, class, constant, continue, 
        constructor, contract, debugger, default, delegatecall, delete, do, else, emit, event, experimental, 
        export, external, false, finally, for, function, gas, if, implements, import, in, indexed, instanceof, 
        interface, internal, is, length, library, log0, log1, log2, log3, log4, memory, modifier, new, payable, 
        pragma, private, protected, public, pure, push, require, return, returns, revert, selfdestruct, send, 
        solidity, storage, struct, suicide, super, switch, then, this, throw, transfer, true, try, typeof, using, 
        value, view, while, with, addmod, ecrecover, keccak256, mulmod, ripemd160, sha256, sha3
    },
    keywordstyle=[1]\color{blue}\bfseries,
    keywords=[2]{
        address, bool, byte, bytes, bytes1, bytes2, bytes3, bytes4, bytes5, bytes6, bytes7, bytes8, bytes9, 
        bytes10, bytes11, bytes12, bytes13, bytes14, bytes15, bytes16, bytes17, bytes18, bytes19, bytes20, 
        bytes21, bytes22, bytes23, bytes24, bytes25, bytes26, bytes27, bytes28, bytes29, bytes30, bytes31, 
        bytes32, enum, int, int8, int16, int24, int32, int40, int48, int56, int64, int72, int80, int88, int96, 
        int104, int112, int120, int128, int136, int144, int152, int160, int168, int176, int184, int192, int200, 
        int208, int216, int224, int232, int240, int248, int256, mapping, string, uint, uint8, uint16, uint24, 
        uint32, uint40, uint48, uint56, uint64, uint72, uint80, uint88, uint96, uint104, uint112, uint120, 
        uint128, uint136, uint144, uint152, uint160, uint168, uint176, uint184, uint192, uint200, uint208, 
        uint216, uint224, uint232, uint240, uint248, uint256, var, void, ether, finney, szabo, wei, days, 
        hours, minutes, seconds, weeks, years
    },
    keywordstyle=[2]\color{teal}\bfseries,
    keywords=[3]{
        block, blockhash, coinbase, difficulty, gaslimit, number, timestamp, msg, data, gas, sender, sig, 
        value, now, tx, gasprice, origin
    },
    keywordstyle=[3]\color{violet}\bfseries,
    identifierstyle=\color{black},
    sensitive=true,
    comment=[l]{//},
    morecomment=[s]{/*}{*/},
    commentstyle=\color{gray}\ttfamily,
    stringstyle=\color{red}\ttfamily,
    morestring=[b]',
    morestring=[b]"
}

\def\BibTeX{{\rm B\kern-.05em{\sc i\kern-.025em b}\kern-.08em
    T\kern-.1667em\lower.7ex\hbox{E}\kern-.125emX}}

\begin{document}

\title{Cross-Chain Options: A Bridgeless, Universal, and Efficient Approach}

\author{
Zifan Peng$^1$  \quad 
Yingjie Xue$^{1,2}$\textsuperscript{†}\thanks{\textsuperscript{†}Corresponding Author: Yingjie Xue}
\quad Jingyu Liu$^1$
\\ 
$^1$Financial Technology Thrust, Hong Kong University of Science and Technology (Guangzhou) \\
$^2$School of Cyber Science and Technology, University of Science and Technology of China \\
zpengao@connect.hkust-gz.edu.cn, yjxue@ustc.edu.cn, jliu514@connect.hkust-gz.edu.cn
}

\maketitle

\begin{abstract}
Options are fundamental to blockchain-based financial services, offering essential tools for risk management and price speculation, which enhance liquidity, flexibility, and market efficiency in decentralized finance (DeFi).
Despite the growing interest in options for blockchain-resident assets, such as cryptocurrencies, current option mechanisms face significant challenges, including a high reliance on trusted third parties, limited asset support, high trading delays, and the requirement for option holders to provide upfront collateral.

In this paper, we present a protocol that addresses the aforementioned issues.
Our protocol is the first to eliminate the need for holders to post collateral when establishing options in trustless service environments (i.e., without a cross-chain bridge), which is achieved by introducing a guarantee from the option writer.
Its universality allows for cross-chain options involving nearly \textit{any} assets on \textit{any} two different blockchains, provided the chains' programming languages can enforce and execute the necessary contract logic.
Another key innovation is reducing option position transfer latency, which uses Double-Authentication-Preventing Signatures (DAPS).
Our evaluation demonstrates that the proposed scheme reduces option transfer latency to less than half of that in existing methods.
Rigorous security analysis proves that our protocol achieves secure option trading, even when facing adversarial behaviors.
\end{abstract}


\section{INTRODUCTION}
\label{sec:intro}
The emergence of blockchains, such as Bitcoin~\cite{bitcoinWhitepaper} and Ethereum~\cite{ethereumWhitepaper}, has given rise to \textit{decentralized finance} (DeFi), transforming traditional financial systems by eliminating the need for intermediaries.
DeFi's open, transparent, and permissionless architecture enables users to access a wide range of financial services, including lending, borrowing, and trading, without being constrained by geographical or institutional barriers.
\textit{Smart contracts} further enhance its appeal by automating transactions, reducing operational costs, and mitigating risks like fraud or human error.
As of January 2025, DeFi protocols hold a total value locked (TVL) of approximately 102 million USD allocated to options markets~\cite{defillama}.

Options are crucial components of DeFi markets, offering participants more sophisticated tools for risk management and investment strategies.
Options provide traders with flexibility in uncertain markets, enabling them to hedge against potential price fluctuations or speculate on future asset values.
An \textit{option} is a bilateral agreement between two parties, where one party holds the right, but not the obligation, to execute a transaction under predetermined conditions.
Assume Alice buys an option from Bob.
Then Alice is called the \textit{holder} and Bob is called the \textit{writer} of the option.
Say the option grants Alice the right to purchase 100 ``guilder'' coins using 100 ``florin'' coins before the weekend.
In exchange for this right, Alice pays Bob 2 florin coins upfront as a \textit{premium}.
If the value of guilder coins rises during the week, Alice can \textit{exercise} her option to acquire them at the lower, agreed-upon price, securing a profit.
However, if the value of guilder coins drops, Alice can choose to \textit{abandon} the option.

There are some existing protocols for trading options on blockchain.
However, these protocols lack \textit{universality}.
Most protocols~\cite{lyrafinance2024, opyn, hegicwhitepaper, aevo, deriprotocol_whitepaper_v4_2023} cannot use assets from any chain as the underlying asset.
They only support options involving specific cryptocurrencies, such as BTC and ETH, on mainstream blockchains like Ethereum~\cite{ethereumWhitepaper}.
These protocols do not allow users to independently create options for non-mainstream cryptocurrencies or other digital assets across all blockchains with customized prices.

To trade options involving any two assets, users can use \textit{cross-chain bridges} to interconnect different chains.
However, current cross-chain bridges are not mature enough for widespread adoption due to numerous security risks and high expenses.
According to statistics~\cite{SlowMist}, cross-chain bridges have experienced at least 47 incidents, resulting in approximately 1.8 billion US dollars being hacked.
Narrowly speaking, cross-chain bridges refer to protocols that enable the transfer of assets or data between different blockchain networks~\cite{cross_chain_security, chainlinkcrosschain, zhangbridges}.
They can be categorized into ``external'' and ``native'' based on their verification methods.\footnote{Some research includes optimistic verification. We categorize optimistic verification as external since it requires at least one honest watcher.}
\textit{External verification} introduces a trusted third party\footnote{The trusted third party here can also refer to Trusted Execution Environments (TEE) or a committee for cross-chain transactions.} to facilitate message transmission. This approach is vulnerable to various attacks, such as rug pulls, private key leakage~\cite{zhangbridges,SlowMist,OrdiZK,polyHack,ALEX}, and software vulnerabilities~\cite{TEESecurity,TEECS,TEEsoftAttack}. 
\textit{Native verification}~\cite{daveasGasEfficientSuperlightBitcoin2020,Glimpse,zamyatinXCLAIMTrustlessInteroperable2019,zkBridge} uses light clients deployed on-chain to verify proofs. One representative solution is zkBridge~\cite{zkBridge}.
However, developing complex light clients for each heterogeneous blockchain is expensive.
The cross-chain verification process is also costly and complicated due to the use of zero-knowledge proofs.

Another widely used technique for cross-chain transactions is the \textbf{Hashed TimeLock Contract} (HTLC)~\cite{AtomicSwaps, Fault-Tolerant}.
HTLCs do not explicitly involve forwarding assets or information from one chain to another for verification, and they only require verification by the parties involved in the transaction.
Parties locally verify the state on one chain and then decide whether to proceed on the other chain.
Researchers~\cite{nieto2018trustminimized,FairnessAtomicSwaps,arwenwhitepaper,liu2020atomic,cfa} have proposed using HTLCs to establish options that effectively address the limitation mentioned above: they enable the trade of any assets across any chain\footnote{Here, we use ``any'' to stress on the wide range of chains. Indeed, it means a blockchain that has a Turing-complete language for smart contracts.}.
HTLCs were initially designed for atomic swaps.
They are shown to be equivalent to a premium-free American call option, which gives holders the right to purchase the underlying asset.~\cite{FairnessAtomicSwaps}.
However, current HTLC-based option protocols face significant challenges regarding collateral and transferability.

For \textbf{collateral}, option traders on two chains must lock their assets as collateral when establishing the option in HTLC.
For example, Alice and Bob are required to lock 100 ``florin'' and 100 ``guilder'' on two chains respectively in the previous example.
Existing protocols require Bob to lock full collateral on his chain, which is reasonable since he bears the risk of fulfilling the obligation, similar to traditional options.
However, it also requires Alice, as the option holder, to provide full collateral upfront.
This is problematic because options are typically much cheaper than their underlying assets and are primarily used by holders for leverage.
To address this, the requirement for upfront collateral should be eliminated.
Eliminating the need for upfront collateral poses non-trivial obstacles.
In HTLC-based options, the option is structured as the holder’s right to exchange their assets for the writer’s assets.
The holder has a key to unlock the writer's collateral, while the writer must wait for the holder to reveal this key to unlock the holder's collateral.
To ensure the writer receives the holder’s assets when the option is exercised, the holder is required to provide collateral.
The writer cannot directly verify if the holder's collateral deposit has actually occurred on another chain in trustless service environments.

Therefore, without the holder's upfront collateral, there is a risk that the writer could give out their assets without receiving anything in return from the holder.
We can possibly utilize cross-chain bridges to solve the problem.
However, as mentioned earlier, current cross-chain bridges present significant risks.
Therefore, our goal is to design a protocol in a trustless service environment to facilitate option traders' transactions.

For \textbf{transferability}, most research focuses on establishing the option contract between the holder and writer.
Once an option is established, it should be tradable since it is valuable.
However, most works overlook this crucial feature.
Previous researchers~\cite{transferable} introduced transferable cross-chain options (short for TCO), allowing options to be transferred post-establishment.
Despite its innovation, their approach has two limitations.
The first one is prolonged transfer time.
In this protocol, the option transfer requires a staggered $9\Delta$ (approximately 9 hours on Bitcoin), where $\Delta$ represents the time period sufficient for each party to release, broadcast, and confirm transactions on the blockchain, typically 1 hour on Bitcoin.
The second one is that it is vulnerable to the \textit{Phantom Bid Attack}, where an adversary creates multiple fake buyers who offer higher prices but fail to finalize the transfer, the option holder/writer is unable to sell their position.

\mypara{Contributions}
Building on existing works, two key challenges remain: removing upfront collateral for holders and enabling efficient and robust option transfers.
In this paper, we introduce a cross-chain option protocol that offers universal accessibility for various assets, enables efficient option transfers, and eliminates the need for upfront holder collateral.
A comparison of other option protocols with ours is shown in \Cref{tab:compare}. Our contributions are summarized as follows:

\begin{table*}[ht]
    \centering
    \caption{Comparison between our protocol and other cross-chain protocols. ``Universal'' represents whether or not parties are allowed to create their own assets on different chains for option trading, and $\Delta$ represents the block confirmation time.}
    \resizebox{\textwidth}{!}{
    \begin{tabular}{c|ccccc|>{\raggedright\arraybackslash}p{7.3cm}}
        \toprule
        \textbf{Protocol} & \textbf{Collateral-free} & \textbf{No External Trust} & \textbf{Universal} & \textbf{Transferable} & \textbf{Transfer Time} & \multicolumn{1}{c}{\textbf{Features}}\\
        \midrule
        Arwen~\cite{arwenwhitepaper} & \color[HTML]{FF4040}\xmark  & \color[HTML]{FF4040}\xmark & \color[HTML]{FF4040}\xmark & \color[HTML]{FF4040}\xmark & - & Off-chain matching and on-chain escrow accounts with HTLC.\\
        \midrule
        Deri V4~\cite{deriprotocol_whitepaper_v4_2023} & \color[HTML]{45C03C}\cmark  & \color[HTML]{FF4040}\xmark & \color[HTML]{FF4040}\xmark & \color[HTML]{FF4040}\xmark & - & Require market making and inputs from an external oracle. \\
        \midrule

        Fernando~\cite{nieto2018trustminimized} & \color[HTML]{FF4040}\xmark & \color[HTML]{45C03C}\cmark & \color[HTML]{45C03C}\cmark & \color[HTML]{FF4040}\xmark & - & HTLC is implemented using Bitcoin script.\\
        \midrule
        
        CFA~\cite{cfa} & \color[HTML]{FF4040}\xmark & \color[HTML]{45C03C}\cmark & \color[HTML]{45C03C}\cmark & \color[HTML]{FF4040}\xmark & - & Compensate to conforming party during option setup phase. \\
        \midrule

        Liu~\cite{liu2020atomic} & \color[HTML]{FF4040}\xmark & \color[HTML]{45C03C}\cmark & \color[HTML]{45C03C}\cmark & \color[HTML]{FF4040}\xmark & - & Allow holder to terminate the option and withdraw collateral.\\
        \midrule
        
        Han \textit{et al.}~\cite{FairnessAtomicSwaps} & \color[HTML]{FF4040}\xmark & \color[HTML]{45C03C}\cmark & \color[HTML]{45C03C}\cmark & \color[HTML]{FF4040}\xmark & - & A method for paying the HTLC-based premium.\\
        \midrule
        
        TCO~\cite{transferable} & \color[HTML]{FF4040}\xmark & \color[HTML]{45C03C}\cmark & \color[HTML]{45C03C}\cmark & \color[HTML]{45C03C}\cmark & $9\Delta$ & HTLC-based cross-chain options position transfer is proposed.\\
        
        \midrule
        \textbf{Our Protocol} & \color[HTML]{45C03C}\cmark & \color[HTML]{45C03C}\cmark & \color[HTML]{45C03C}\cmark & \color[HTML]{45C03C}\cmark & $4\Delta$ & \textbf{Fully trustless, transferable, and efficient.} \\
        \bottomrule
    \end{tabular}
    }
    \label{tab:compare}
\end{table*}

\begin{itemize}
    \item \textbf{Efficient and Robust Option Transfer.}
    We propose an option transfer protocol that enables the transfer of assets with less latency and more resilience to phantom bid attacks.
    The key idea is the adoption of \textbf{Double-Authentication-Preventing Signatures} (DAPS) to option transfer finalization.
    By DAPS, our proposed protocol completes the option transfer more than twice as fast as state-of-the-art methods when all parties are honest.
    It can defend against the phantom bid attacks by finalizing the option transfer solely by the option seller.
    
    \item \textbf{Holder Collateral-Free Cross-Chain Options}.
    We design an option protocol without requiring the holder to post collateral when the option is established.
    The key idea is the introduction of economic incentives.
    We let the option writer provide a fair \textit{guarantee} before the option is established.
    The guarantee serves as a compensation to the holder if the writer fails to fulfill the obligation when the holder exercises the option.
    
    \item Combining the above two schemes, we propose a protocol for efficient cross-chain options without upfront collateral from holders, achieving an option trading protocol closer to those of traditional options.
    What's more, the protocol can support concurrent bidding of multiple potential buyers.
    The option transfer time is constant regardless of the number of bidders.
    
    \item We provide a rigorous security analysis and implementation of our proposed protocol. 
\end{itemize}

The rest of the paper is outlined as follows.
\Cref{sec:model} introduces the preliminaries.
\Cref{sec:protocol} provides an overview of the holder collateral-free cross-chain option and the efficient option transfer protocol.
Detailed descriptions of the protocols are presented in \Cref{sec:details}.
Security properties are analyzed and proven in \Cref{sec:analysis}.
We present implementation and evaluation details in \Cref{sec:impl}.
Related works are discussed in \Cref{sec:related}.
\Cref{sec:conclusion} summarizes this paper.

\section{PRELIMINARIES}
\label{sec:model}
\mypara{Blockchain, Smart Contract, and Asset}
A \textit{blockchain} is a tamper-proof distributed ledger that records asset balances for each address.
An \textit{asset} can be a cryptocurrency, a token, or any item transactable on-chain.
A \textit{party} can be an individual, organization, or any entity capable of interacting with the blockchain.
A \textit{smart contract} (simply as ``contract'') is an immutable agreement written in code and executed by a blockchain system.
Parties can create contracts, call functions, and check contract code and state.
$\Delta$ represents the sufficient time period for the parties to release, broadcast, and confirm transactions on the blockchain.

\mypara{Cryptographic Primitives}
A secret is a fixed-length sequence of bits.
A secret is known exclusively to its generator, and $H(\cdot)$ represents a collision-resistant hash function.
In asymmetric encryption, the private key $sk$ is used confidentially for signing, while the public key $pk$ is shared openly.

In this paper, Double-Authentication-Preventing Signatures (DAPS)~\cite{DAPS} is a key component in our protocol design.
Initially, DAPS is designed to inhibit the reuse of a single private key for signing two different messages, where a message consists of a pair of message address and message payload in the form of $(a, p)$.
DAPS ensures that a particular secret key \( sk \) cannot sign the same address \( a \) with different payloads \( p \). This property can be used to prevent double spending in blockchains.
Two messages $m_1 = (a_1, p_1)$ and $m_2 = (a_2, p_2)$ are considered colliding if $a_1 = a_2$ and $p_1 \neq p_2$.
Any two signatures with the identical address but different contents will disclose the secret key. Given a security parameter $\lambda$, DAPS can be delineated as follows.

\begin{itemize}
    \item \textbf{Key Generation:} \(\text{KeyGen}(1^\lambda)  \rightarrow (pk, sk)\)
    \item \textbf{Signature:} \(\text{Sign}(sk, m)  \rightarrow \sigma_m, \text{ where } m = (a, p)\)
    \item \textbf{Verification:} \(\text{Verify}(pk, m, \sigma_m) \rightarrow \text{True/False}\)
    \item \textbf{Extraction:} \(\text{Extract}(pk, m_1, \sigma_{m_1}, m_2, \sigma_{m_2}) \rightarrow sk/\bot\)
\end{itemize}

\mypara{Hashed TimeLock Contracts (HTLCs)}
The Hashed TimeLock Contract (HTLC) is a cryptographic contract utilized to facilitate secure and trustless transactions. In a vanilla HTLC, funds are locked in a contract and can only be accessed by the designated recipient upon fulfillment of predetermined conditions within a specified time frame $T$. The condition is expressed as the presentation of a preimage of the hash. For example, the contract asks the designated recipient to present the preimage $A$ for the hash $H(A)$. If $A$ is not provided before $T$, the funds are refunded after $T$.

\section{OVERVIEW}
\label{sec:protocol}
Consider the scenario described in \Cref{sec:intro}.
Alice and Bob strike a deal: Bob agrees to sell his 100 guilder coins to Alice in exchange for 100 florin coins within one week.
To secure this arrangement, Alice pays a premium with another kind of asset, e.g., 1 US dollar, for the option to complete the transaction.
In the world of blockchain, the 100 florin and 100 guilder coins are referred to as $Asset_A$ on $Chain_A$ and $Asset_B$ on $Chain_B$, respectively, with one week as the expiration time $T_E$.
A 1 US dollar premium could be an asset $P$ on another $Chain_P$, including $Chain_A$ and $Chain_B$.
An HTLC-based option requires Alice to pay a premium $P$ and escrow $Asset_A$ at the beginning of the option transaction, followed by Bob escrowing $Asset_B$.

Since this option itself is valuable, others may want to make transactions with the holder or writer.
Suppose there are $L$ holder position bidders, $Carol_i$, and $M$ writer position bidders, $Dave_j$, where $i \in \{1, 2, \ldots, L\}$ and $j \in \{1, 2, \ldots, M\}$.
They are willing to pay a \textit{holder transfer fee}, $\texttt{HolderFee}_i$, and a \textit{writer transfer fee}, $\texttt{WriterFee}_j$, to obtain the option position from Alice and Bob, respectively.
The holder transfer fee represents the price for the option and Alice's asset locked in the contract, while the writer transfer fee is the fee to acquire Bob's risky asset $Asset_B$ locked on chain with obligations of this option tied to that asset.

\mypara{Making transfer more robust and efficient}
Consider an HTLC-based option where Alice and Bob lock their assets with the hashlock \(H(A)\), where \(A\) is the \textit{exercise secret} generated by Alice.
Consider the case where Carol purchases Alice's position.
In the previous work (TCO)~\cite{transferable}, the protocol first locks the contracts by tentatively assigning a new hashlock to replace the old one.
Alice can place different hashlocks on two chains, and Carol may replace Alice on one chain but not another.
Bob is given a time-consuming consistency phase to ensure the new hashlocks are consistent.
We would like to reduce the time needed for Alice's position transfer.

We manage to reduce the transfer time by a key observation.
Since Alice cannot transfer the option to multiple bidders simultaneously, it logically prompts the use of Double-Authentication-Preventing Signatures (DAPS) to prevent a seller from selling signatures to multiple bidders, thus removing the requirement of guaranteeing consistency by Bob's efforts.
If multiple signatures are revealed, a secret key can be extracted by DAPS, and then punishment is enforced automatically to ensure a fair payoff. With DAPS, the transaction completion in our protocol is less than half the time of TCO~\cite{transferable}.

\mypara{Holder Collateral-Free Cross-Chain Options}
Our objective is to allow Alice to pay a premium to secure this right, without the need to escrow the assets in advance.
We need a mechanism to enable the correct exercise of this right \textemdash Alice pays her florin coins to get Bob's guilder coins.
Alice cannot get Bob's coins without paying florin coins to Bob.
A naive approach would be to require cross-chain transaction confirmation of Alice's escrow when Alice decides to exercise.
However, every blockchain cannot directly get the cross-chain transaction confirmation information from another chain.
This approach requires cross-chain bridges, which suffer from various security issues~\cite{sokHacks,SlowMist}, such as key leakage~\cite{ALEX}, smart contract vulnerabilities~\cite{polyHack,ronin}, and rug pulls~\cite{5RugPulls,OrdiZK}.
The goal is to design a holder collateral-free option without a trusted cross-chain bridge.

If we grant Alice direct access to the exercise right (exercise secret in HTLC), then Bob's interests are not protected, as Alice could withdraw Bob's coins directly.
To resolve this problem, we resort to economic incentives that are commonly present in the DeFi markets, which are also driving forces for options.
We let Bob control the exercise secret while Alice retains the right to penalize Bob.
In addition to a collateral (100 guilder coins in our example) required by normal option contracts, we ask Bob to lock another valuable asset on $Chain_A$ as a guarantee for Alice when she wants to exercise the right.
If Alice wants to exercise her right after 2 days, she locks her asset on $Chain_A$ (100 florin coins in our example).
If Bob does not release the exercise secret, Alice will get Bob's guarantee as compensation.
Suppose this guarantee is sufficient to compensate for the expected profit of this option.
This method provides Bob with incentives to cooperate when Alice wants to exercise her right.

\mypara{Discussion on the Guarantee}
Considering the value of guarantee, i.e., $Asset_G$, options trading occurs due to differing asset expectations.
Alice buys options for speculation or insurance, while Bob earns a premium.
This option ensures that Alice gets at least $Asset_G$ upon exercising.
For rational Alice, if her expected return and the agreed $Asset_G$ exceed the premium, she will purchase the option.
Typically, this premium is about 2\%~\cite{FairnessAtomicSwaps} of the underlying asset.
Therefore, the guarantee addition by $Asset_G$ is insignificant.
Our work focuses on the mechanism setup, leaving precise pricing for future work.

Additionally, $Asset_G$ can be any asset agreed upon by Alice and Bob, such as stablecoins to mitigate potential risks, highly volatile assets, or assets with low liquidity, like digital collections (e.g., NFTs). Thus, this protocol is attractive to writers with limited liquidity.
Depending on Alice's risk tolerance, she can choose whether to accept a volatile $Asset_G$ as a guarantee.
A risk-averse Alice should opt for a stable-valued $Asset_G$.
In summary, the choice of guarantee comes from the traders' negotiation and willingness.
Our proposed protocol provides trading opportunities for holders and writers who accept the risks.

\mypara{Integration} 
We integrate the efficient option transfer protocol into the holder collateral-free cross-chain options.
Bob controls the exercise secret and transfer process instead of Alice, the processes for transferring the holder and writer are reversed. 
The hash of the exercise secret on the chains must remain identical after Bob transfers his position to Dave.
A key challenge is ensuring that honest parties incur no loss by addressing potential misbehavior by any party and collusion between any two parties.
To counter these, we introduce a withdrawal delay to allow Dave to retrieve assets in case Bob acts maliciously and a transfer confirmation delay to allow Alice to contest any inconsistent replacement of the hashlock. 

\section{PROTOCOL}
\label{sec:details}
Due to the complexity of efficient holder collateral-free options, we will elaborate on the protocol gradually. We first introduce the efficient transfer of an option in \Cref{sec:ft-htlc}. Next, we outline how to achieve holder collateral-free cross-chain options in \Cref{sec:cfop}. Finally, we show the efficient, holder collateral-free option protocol.

\subsection{Efficient Option Transfer Protocol}
\label{sec:ft-htlc}
\mypara{Option Initialization} Firstly, we illustrate an efficient option transfer protocol in an HTLC-based option. Assume Alice and Bob initialize an HTLC-based option as the holder and writer, respectively.
In this option, Alice locks $Asset_A$ on $Chain_A$, intending to transfer it to Bob if a preimage of $H(A)$ is presented before $T_E+\Delta$.
Bob locks $Asset_B$ on  $Chain_B$, intending to transfer it to Alice if a preimage of $H(A)$ is presented before $T_E$, where $T_E$ is the expiration time of this option.
Alice owns the preimage $A$.

In addition, Alice performs $\text{KeyGen}(1^\lambda) \rightarrow (pk_A, sk_A)$, which acts as a \textit{transfer key} pair. This transfer key pair is used as a DAPS key pair for misbehavior detection, where $pk_A$ is recorded in both contracts.
The transfer key is used by Alice when transferring ownership to another party.
A signature generated by $sk_A$ can be used to replace the contract holder, the hashlock, and the new transfer public key. Similarly, Bob creates a transfer key and records it on chains.
Alice and Bob agree in advance on a value (e.g., a 256-bit random number) to serve as the message address $a$ recorded in the contracts for the DAPS.
We take the holder position transfer as an example to illustrate this transfer protocol.

\subsubsection{Transfer Holder's Position} 
\label{sec:ft-holder}
Suppose Alice reaches an agreement with Carol to transfer the holder position on or before time $T_H$, with a charge of $\texttt{HolderFee}$ on $Chain_C$.
Carol performs $\text{KeyGen}(1^\lambda) \rightarrow (pk_C, sk_C)$ to generate a new transfer key pair. Carol deposits $\texttt{HolderFee}$ in $Contract_C$.
This contract requires a signature of \sloppy $m=(a, p)$, where message payload $p = (\text{Carol}.address, H(C), pk_C)$, signed by $sk_A$ to unlock and transfer the $\texttt{HolderFee}$ to Alice.
$H(C)$ represents a new hashlock chosen by Carol, which is used to exercise the option if Carol obtains the holder position.
Besides, $Contract_C$ records $H(A)$, specifying that $\texttt{HolderFee}$ is refunded to Carol if Carol can reveal $A$ (meaning that Alice has exercised the option).
Instead of withdrawing immediately, after Alice reveals a signature to redeem $\texttt{HolderFee}$, she must wait for $3\Delta$ to elapse. We refer to this period as the \textit{Withdrawal Delay Period}.
The protocol consists of two phases, \Cref{fig:transfer_holder} illustrates the position transferring process:

\begin{enumerate}
    \item \textbf{Reveal Phase}: Carol locks the transfer fee and Alice attempts to withdraw the transfer fee with her signature.
    \item \textbf{Consistency Phase}: Carol forwards the signature to replace the holder and Alice withdraws the transfer fee after the withdrawal delay period.
\end{enumerate}

\noindent I. \textbf{Reveal Phase.}
\begin{enumerate}
    \item Alice generates a signature by $\text{Sign}(sk_A, m) \rightarrow \sigma_m$, where $m$ equals to $(a, (\text{Carol}.address,H(C),pk_C))$.
    \item If Alice wants to transfer option to Carol, Alice sends $\sigma_m$ in $\textit{Contract}_C$ by invoking the function \texttt{reveal()} and wait for $3\Delta$ (withdrawal delay period). If she abandons this transfer, she does not reveal $\sigma_m$. The $\texttt{HolderFee}$ will be refunded to Carol after $T_H$.
\end{enumerate}

\noindent II. \textbf{Consistency Phase.}
\begin{enumerate}
    \item Carol\footnote{Any party can forward this signature, as Alice may transfer to any party.} forwards $\sigma_m$ to both $Contract_A$ and $Contract_B$ directly, attempting to call the function \texttt{transferHolder()} to replace the holder to Carol, the hashlock to $H(C)$, and transfer public key to $pk_C$.
    \item Alice calls \texttt{withdraw()} to obtain the $\texttt{HolderFee}$ in $Contract_C$ after the withdrawal delay period.
\end{enumerate}

If all parties perform honestly, Alice is able to receive $\texttt{HolderFee}$ and the holder is changed to Carol. However, there are possible contingent events or dishonest scenarios:

\begin{itemize}
    \item If Alice exercises the option during the transfer process and reveals the preimage $A$ before $T_H$, Carol can refund the $\texttt{HolderFee}$ from $Contract_C$ using $A$ during the withdrawal delay period.
    \item If different signatures with the same message address $\sigma_{m'} \neq \sigma_m$, are submitted on $Chain_A$ and $Chain_B$ (e.g., if Alice submits two different signatures or sells the option to multiple parties), anyone can call $\text{Extract}(pk, m', \sigma_{m'}, m, \sigma_m) \rightarrow sk_A$ to get $sk_A$. $sk_A$ is the secret transfer key of Alice. Whoever gets $sk_A$ means that Alice misbehaves. We can use this as evidence for fair settlement of funds.
    \begin{itemize}
        \item Carol can call \texttt{reclaim()} and obtain $\texttt{HolderFee}$ with $sk_A$ during withdrawal delay period.
        \item Bob can use $sk_A$ to claim both $Asset_A$ and $Asset_B$. If a signature has not been submitted, Bob can claim it anytime.
        If a signature has been submitted, Bob sends $sk_A$ within one $\Delta$ after the signature submission.
    \end{itemize}
    \item If Carol reveals $\sigma_m$ on only $Contract_A$ or $Contract_B$, Bob can forward the signature to the other contract.
\end{itemize}

\textit{Timeouts.}
The transfer contract must be created no later than $T_H-3\Delta$, and the reveal phase should be completed by $T_H-2\Delta$ to ensure that the option can be transferred to Carol at $T_H$.
In the consistency phase, if any misbehavior occurs, it should be reported to the contract by $T_H+\Delta$.
If Bob does not claim assets on $Chain_A$ and $Chain_B$ with $sk_A$, then it implies transfer is complete.
Overall, a total transfer time of $4\Delta$ is required.
In other words, the transfer protocol must initiate no later than $T_E - 4\Delta$. The unlocking condition for $Contract_C$ is summarized in \Cref{tab:unlock_condition}.

\begin{table}[ht]
    \centering
    \caption{Unlocking conditions of transferring position from Alice to Carol, where $T$ and $T_R$ are current time and the time that Alice reveals the signature.}
    \begin{tabular}{c|c}
        \toprule
        \textbf{Recipient} & \textbf{Condition} \\
        \midrule
        Alice & $\text{Sign}(sk_A, (p_C,H(C),pk_C) \land T > T_R + 3\Delta$\\
        \midrule
        \multirow{3}*{Carol} & $sk_A \lor A$ \\
        ~ & $\text{Sign}(sk, m') \land m' \neq m$ \\
        ~ & $T \geq T_H-\Delta \land T_R = \bot $ \\
        \bottomrule
    \end{tabular}
    \label{tab:unlock_condition}
\end{table}

\begin{figure}
  \centering
  \includegraphics[width=\linewidth]{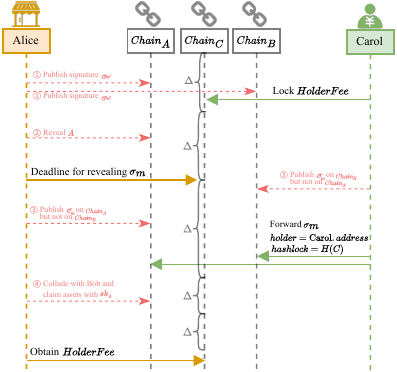}
  \caption{Alice transfers the holder's position to Carol. The red dashed lines represent malicious activities, illustrated in \Cref{sec:misbehavior}.}
  \label{fig:transfer_holder}
\end{figure}

\subsubsection{How misbehaviors are handled securely in the protocol}
\label{sec:misbehavior}
Here we show how this protocol handles misbehaviors and protects each party's interests by ensuring a fair payoff for honest parties in red dashed lines of \Cref{fig:transfer_holder}.
A more rigorous analysis is shown in the full version~\cite{peng_full_version} of this paper.
First, we consider each party acting maliciously on its own. 
\begin{itemize}
    \item If Alice provides two different signatures to different buyers, as shown in \textbf{\textcircled{1}}, Bob can extract $sk_A$ and submit it to obtain $Asset_A$ and $Asset_B$, and Carol can reclaim the transfer fee with $sk_A$.
    In that case, Bob does not lose his $Asset_B$, and Carol does not lose her transfer fee.
    \item If Alice reveals $A$ at the same time during the transfer process, as shown in \textbf{\textcircled{2}}, Carol can use $A$ to reclaim $\texttt{HolderFee}$.
    She does not lose anything. The option is exercised, and a swap happens between Alice and Bob.
    \item If Alice or Carol publishes one signature exclusively on either $Contract_A$ or $Contract_B$, as shown in \textbf{\textcircled{3}}, Bob can forward this signature to another chain to make sure the hashlocks and holders are consistent on two chains.
\end{itemize}

Next, we consider scenarios where collusion exists.
\begin{itemize}
    \item If Alice and Bob collude, they can use $sk_A$ or $A$ to withdraw $Asset_A$ and $Asset_B$ as shown in \textbf{\textcircled{4}}. Carol can observe $sk_A$ or $A$ and withdraw $\texttt{HolderFee}$ during the withdrawal delay period.
    \item If Alice and Carol collude, they use two signatures to change the hashlock. During the withdrawal delay period, Bob can obtain $Asset_A$ and $Asset_B$ using the extracted $sk_A$, which is reduced to \textbf{\textcircled{1}}.
    \item If Bob and Carol collude, they cannot do anything harm. Since Alice will only reveal one valid signature, Alice will receive $\texttt{HolderFee}$ from Carol.
\end{itemize}

\subsubsection{Transfer Writer's Position}
\label{sec:ft-writer}
Transferring the writer's position is similar but simpler because Bob does not possess the preimage of the hashlock. Bob, with the transfer key pair $(pk_B, sk_B)$, can sign the message $m = (a, (\text{Dave}.address, pk_D))$ using $sk_B$ to collect the transfer fee, where $pk_D$ is a new transfer key for Dave. Transferring the writer's position does not update the hashlock used in the option exercise. Thus, Alice's option is not influenced except for the change of the new option writer.

\subsection{Holder Collateral-Free Cross-Chain Options}
\label{sec:cfop}

\begin{figure}
  \centering
  \includegraphics[width=\linewidth]{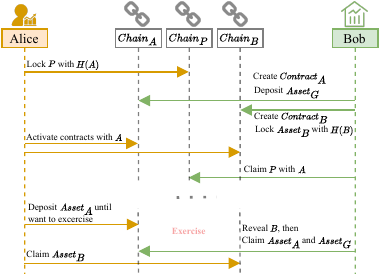}
  \caption{Collateral-Free Cross-Chain Swap Options if both Alice and Bob are honest, where Alice generates $A$ and $H(A)$, Bob generates $B$ and $H(B)$, then they exchange $H(A)$ and $H(B)$.}
  \label{fig:option_honest}
\end{figure}

We want to remove the need for upfront collateral from Alice without using a cross-chain bridge. Allowing Alice direct access to the exercise secret risks Bob's asset since Alice has no collateral. To address this, we resort to economic incentives and let Bob control the exercise secret while Alice retains the right to penalize Bob. In addition to the usual collateral, Bob locks a valuable asset on $Chain_A$ as a guarantee. If Bob fails to release the exercise secret when Alice exercises her right, she receives Bob's guarantee as compensation, incentivizing Bob to cooperate.

Suppose Alice and Bob reach an agreement that Alice pays Bob $P$ as a premium on a chain denoted by $Chain_P$. The option takes effect at $T_A$ meaning that Alice obtains the right to exchange $Asset_A$ for $Asset_B$ before $T_E$. Bob's guarantee is $Asset_G$.
The protocol involves two kinds of asset settlement: first for option establishment (or activation, we use them interchangeably) and second for option exercise. We therefore introduce two secrets: (1) \textit{Activation} secret $A$, used for Alice to pay the premium and activate the option; and (2) \textit{Exercise} secret $B$, used for Alice to pay $Asset_A$ to Bob in exchange for \( Asset_B \) when the option is exercised. Secret $A$ is generated by Alice and $B$ is generated by Bob.

The protocol is divided into two phases. \Cref{fig:option_honest} shows the execution process of this protocol if both parties are honest.
\begin{enumerate}
    \item \textbf{Setup phase}: Alice and Bob activate an option. Alice obtains option and Bob obtains premium.
    \item \textbf{Exercise/Abandon phase}: Alice can either exercise the option or abandon it.
\end{enumerate}

In the setup phase, Alice and Bob will establish this option similarly to a vanilla HTLC.
Alice locks $P$ with a hashlock $H(A)$ in a contract on any chain.
Bob creates two contracts, $Contract_A$ on $Chain_A$ and $Contract_B$ on $Chain_B$, which are used in the option.
The option remains inactive until Alice reveals the activation secret $A$ before $T_A$, at which point the state updates to active and Bob gets Alice's premium.
$Contract_A$ holds Bob's guarantee, $Asset_G$, until the option expires.
If Alice exercises the option and Bob fulfills his obligation by revealing the exercise secret $B$, $Asset_G$ is refunded to Bob.
If Bob fails to fulfill his obligation, $Asset_G$ will be transferred to Alice.
$Contract_B$ holds Bob's collateral, $Asset_B$, with hashlock $H(B)$.

\noindent I. \textbf{Setup Phase.}
\begin{enumerate}
    \item Alice randomly selects a secret $A$ as the activation secret, and computes its hash value $H(A)$. Bob generates $B$ and $H(B)$, which serve as the exercise secret and hashlock.
    \item Alice locks $P$ with hashlock $H(A)$ on the agreed-upon $Chain_P$ with a timeout $T_A+\Delta$.
    \item If Bob observes that Alice has honestly deposited the premium, Bob should, at any time before $T_A-\Delta$:
    \begin{enumerate}
        \item Create $Contract_A$ on $Chain_A$ and $Contract_B$ on $Chain_B$. These contracts are initially in an inactive state, and record the holder and writer, activation time $T_A$ and option expiration time $T_E$.
        \item Escrow the guarantee $Asset_G$ on $Chain_A$, and lock principal $Asset_B$ on $Chain_B$ with hashlock $H(B)$.
    \end{enumerate}
    \item If Alice observes that Bob has created contracts and made deposits, Alice reveals $A$ at $T_A$ on both chains to activate the option. If not, transaction aborts, Bob calls \texttt{refund()} and retrieves $Asset_G$ and $Asset_B$. Alice refunds the premium $P$.
\end{enumerate}

\noindent II. \textbf{Exercise/Abandon Phase.}
\begin{enumerate}
    \item \textit{Exercise}: If Alice wants to exercise the option at $T_B$ before expiration, she calls \texttt{exercise()} and deposits $Asset_A$ into $Contract_A$, then within one $\Delta$:
    \begin{enumerate}
        \item If Bob reveals $B$ on $Contract_A$, then he obtains both $Asset_A$ and $Asset_G$. Upon observing $B$, Alice obtains $Asset_B$ with $B$ from $Contract_B$.
        \item If Bob does not reveal $B$, Alice calls \texttt{claim()} on $Contract_A$ after $T_B+\Delta$ to receive $Asset_G$ as compensation.
    \end{enumerate}
    \item \textit{Abandon}: If Alice does not call \texttt{exercise()} before or at $T_E$, then the option is abandoned and Bob can call \texttt{refund()} on $Contract_A$ and $Contract_B$ to refund $Asset_G$ and $Asset_B$.
\end{enumerate}

\textit{Timeouts.}
The latest deadline $T_B$ is no later than $T_E$. If Bob fails to fulfill his obligations, then Alice receives $Asset_G$ by $T_E+2\Delta$.
Therefore, the lock period for $Asset_G$ in $Contract_A$ is $T_E+\Delta$ if Alice waives the option, or extends to $T_E+2\Delta$ if Alice exercises the option.
Alice exercises the option and receives \(Asset_B\) by \(T_E + 2\Delta\). Therefore, the lock period for \(Asset_B\) in $Contract_B$ is \(T_E + 2\Delta\).

\subsubsection{Integration: Efficient Cross-Chain Options without Upfront Holder Collateral}
\label{sec:ftcop}
We incorporate the efficient option transfer protocol to enable a collateral-free option transfer process. From the option transfer perspective, the roles of the holder and writer are reversed, as Bob owns the exercise secret. Bob deposits $Asset_G$ and $Asset_B$ in $Contract_A$ and $Contract_B$. In the transfer of Bob's position, hashlock $H(B)$ must remain consistent.

Take Bob transferring the writer's position to Dave as an example. It is similar to the \Cref{sec:ft-holder} with three notable differences. Suppose Bob reaches an agreement with Dave to transfer the writer position.
Dave is able to buy Bob's risky asset with its obligation at the price of $\texttt{WriterFee}$ before or at $T_W$.
Firstly, Dave should choose a new hashlock as the exercise secret.
Similarly, Bob needs to use his private key $sk_B$ to sign Dave's new hashlock $H(D)$, which means message $m = (a, (\text{Dave}.address, H(D), pk_D))$.
Secondly, Alice can use Bob's private key $sk_B$ to reclaim $Asset_B$ and guarantee, $Asset_G$.
Thirdly, if Alice wants to exercise the option and makes the deposit after Bob reveals the signature during the transfer process, the transfer continues, and Dave should forward the signature to obtain the writer's position.
Dave should fulfill his obligation and reveal the exercise secret at $T_W - \Delta$ on $Contract_A$.

\subsubsection{Support Concurrent Bidding and Defending Against Phantom Bid Attacks}
\label{sec:concurrent}
One notable feature of our proposed protocol is its ability to support concurrent bidding without increasing option transfer time.
For example, in a holder option transfer, assume there are $L$ holder position bidders, denoted as $Carol_i$, each willing to pay $\texttt{HolderFee}_i$ to obtain Alice's option position.
Each $Carol_i$ creates a contract, $Contract_C^i$, on $Chain_C^i$ and locks $\texttt{HolderFee}_i$ until Alice reveals her signature.
Alice can select one bidder $Carol_i$ to trade with and reveal her signature in $Contract_C^i$.
She will then receive the transfer fee she chooses and will not be subject to a fraudulent transaction that cannot be finalized.
Any option transfer is settled within a constant time window, regardless of the number of bidders.

Since the support for concurrent bidding, our protocol can effectively defend against \textit{phantom bid attack}. In the phantom bid attack, an adversary creates multiple virtual buyers who offer higher prices but do not finalize the transfer.
In TCO~\cite{transferable}, which attempts to transfer the option to a buyer sequentially, in the face of such an attack, the option holder/writer cannot sell their positions in a reasonable time since the virtual buyers are exhausting the option transfer window. 

With our proposed protocol, an adversary option buyer cannot launch this attack. This is due to the use of a signature for option transfer settlement, rather than a hashlock used in the previous protocol.
By this signature scheme, once a buyer is chosen by the seller, the option transfer can be finalized. There is no time window for the buyer to choose to finalize the option transfer or abort.  

\section{SECURITY ANALYSIS}
\label{sec:analysis} 
In this section, we analyze and prove the properties satisfied by our proposed protocol.
Firstly, recall from the previous illustration that Alice/Carol is called the \textit{holder}, and Bob/Dave is the \textit{writer}, and Carol or Dave is required to pay a \textit{transfer fee} to the corresponding seller.
The payment $P$ is called the \textit{premium}. $Asset_G$ deposited by Bob is called \textit{guarantee}. \textit{transfer fee} or any other asset escrowed in contracts is referred to as \textit{collateral}.

\subsection{Option Transfer Properties}
There are two fundamental properties that must be satisfied during the option transfer process.
\begin{itemize}
    \item \textit{Safety (No Underwater)}: A compliant party must not lose any collateral/position during the transfer process without acquiring a corresponding collateral/position.
    \item \textit{Liveness}: During the transfer of the holder or writer position, if all parties are conforming, the holder or writer will correctly transfer to the buyer, and the holder or writer will receive the correct transfer fee.
\end{itemize}

\begin{restatable}{theorem}{safety}
    \label{theorem:safety}
    Protocol \ref{sec:ftcop} satisfies safety (no-underwater):
    \begin{itemize}
        \item If Alice is conforming and loses her position, then she will obtain either Carol's collateral, Bob's collateral, and guarantee, or both.
        \item If Bob is conforming and loses his position, then he will obtain Dave's collateral.
        \item If Carol is conforming and loses her collateral, then Carol will obtain the holder position of the option.
        \item If Dave is conforming and loses his collateral, then Dave will obtain the writer position of the option.
    \end{itemize}
\end{restatable}

\begin{proof}
See details in \Cref{sec:safety}.
\end{proof}

\begin{lemma}
\label{lemma:parti}
The holder transfer procedure of Protocol \ref{sec:ftcop} does not require Bob's participation.
\end{lemma}

\begin{proof}
Evidently, Alice does not own the exercise secret, holder's transfer is required to replace the holder address and transfer public key, and the inconsistency of two chains will not harm the interest of Bob.
According to the Protocol \ref{sec:ft-htlc}, Bob cannot use the transfer private key of Alice, i.e. $sk_A$ to claim assets.
Therefore, during the reveal phase and consistency phase, Bob is not required to participate and is not allowed to make any changes on $Contract_A$ and $Contract_B$.
\end{proof}

\begin{lemma}
\label{lemma:alice_parti}
If Bob and Dave are conforming, then the writer transfer procedure of Protocol \ref{sec:ftcop} does not require Alice's participation.
\end{lemma}

\begin{proof}
Obviously, honest Bob will not leak two signatures or $sk_B$ and honest Dave will submit signature $\sigma_m$ on both $Contract_A$ and $Contract_B$, Alice only needs to make operations when there is any dishonest party.
\end{proof}

\begin{restatable}{theorem}{correctness}
\label{theorem:correctness}
Protocol \ref{sec:ftcop} satisfies liveness: If Alice, Bob, and Carol/Dave are conforming, then Alice/Bob will obtain Carol/Dave's collateral, Carol/Dave will obtain Alice/Bob's position, and Bob/Alice will retain their original position.
\end{restatable}

\begin{proof}
By \Cref{lemma:parti}, Bob's participation is not required during the holder transfer.
If Alice and Carol are conforming, Carol creates $Contract_C$ contract and lock her collateral using the signature of Alice before $T_H-3\Delta$.
Alice will then reveal the signature by $sk_A$ and call $reveal()$ on $Contract_C$ at $T_H-2\Delta$. An honest Carol will forward the signature, setting the holder to Carol.
Alice can then wait for $3\Delta$ withdrawal delayed period to obtain the collateral, while the writer of $Contract_A$ and $Contract_B$ is still Bob, Bob maintains the writer's position.
During the process where Bob transfers his position to Dave, if both parties are conforming, Bob will not expose two different signatures.
After $T_W + \Delta$, Bob will not be obstructed and will obtain Dave's collateral.
Meanwhile, Dave can submit $\sigma_m$ between $T_W-\Delta$ and $T_W$ to $Contract_A$ and $Contract_B$ to change the writer, Alice retains the holder position.
\end{proof}

\Cref{theorem:safety} guarantees that an honest party will not incur losses in the transfer protocol, even if another party irrationally forfeits their assets, thus protecting the interests of contract adherents.
\Cref{theorem:correctness} ensures that if the transfer proceeds correctly, all participants receive the intended outcomes.

Another three properties must be ensured during the transfer process: \textit{Unobstructibility}, \textit{Independence}, and \textit{Isolation}.
First, external interferences cannot disrupt the transfer by the conforming parties. Second, after a successful transfer, the holder or writer is updated to another party, ensuring that the transfer does not affect future transactions.
Finally, the transfer process for the holder and the writer must be isolated, with no mutual influence.
We list theorems related to other properties, with proofs in \Cref{sec:trans_proofs}. 

\begin{itemize}
    \item \textit{Unobstructibility}: If both the buyer and seller are conforming, no other party can obstruct the transfer process.
    \item \textit{Independence}: After a successful transfer, subsequent transfers can proceed normally, and the previous position holder cannot interfere with future transfers.
    \item \textit{Isolation}: The processes of transferring options by the holder and writer to buyers can occur simultaneously and separately, without interference from each other.
\end{itemize}

\begin{restatable}{theorem}{unobstructibility}
\label{theorem:unobstructibility}
Protocol \ref{sec:ftcop} satisfies transfer unobstructibility: Alice/Bob can transfer the position to another party even if Bob/Alice is adversarial.
\end{restatable}

\begin{restatable}{theorem}{sbindependence}
Protocol \ref{sec:ftcop} satisfies transfer independence: After $Alice_i$/$Bob_j$ transfers to $Alice_{i+1}$/$Bob_{j+1}$, the new position owner can subsequently transfer to $Alice_{i+2}$/$Bob_{j+2}$ without interference from any adversarial party.
\end{restatable}

\begin{restatable}{theorem}{counterparty}
\label{theorem:counterparty}
Protocol \ref{sec:concurrent} satisfies transfer isolation: Alice and Bob can simultaneously and separately transfer their positions to Carol and Dave, respectively. This means that the transferring holder and the transferring writer can proceed concurrently.
\end{restatable}

\subsection{Option Properties}
In addition to the properties of the protocol during the option transfer process, we explore the properties of option contracts.

\begin{itemize}
    \item \textit{Option correctness}: If both the holder and writer are conforming, either the exercise does not occur, the holder does not lose their collateral and the writer does not lose their collateral and guarantee; or upon completion of the exercise, they will each receive the other's collateral, and the writer will reclaim their guarantee.
    
    \item \textit{Exercisablity}: During the transfer of the writer position, the option holder can exercise the option without experiencing any delays or obstructions.

    \item \textit{Failure compensation}: If the holder initiates the exercise before expiration, he will either successfully exercise the option or receive the pre-agreed compensation guarantee.
\end{itemize}

\begin{restatable}{theorem}{opcor}
\label{theorem:opcor}
Protocol \ref{sec:cfop} satisfies option correctness: If both Alice and Bob are conforming, then if Alice does not exercise the right, Alice doesn't lose $Asset_A$ and Bob doesn't lose $Asset_G$ and $Asset_B$; or if Alice exercises the right, then Alice will receive $Asset_B$ and Bob will receive $Asset_A$ and $Asset_G$.
\end{restatable}

\begin{restatable}{theorem}{liveness}
Protocol \ref{sec:cfop} satisfies exercisablity: During the transfer from Bob to Dave, the option remains active, allowing Alice to exercise the option without any delays.
\end{restatable}

\begin{restatable}{theorem}{compensation}
Protocol \ref{sec:cfop} satisfies failure compensation: Before expiration, Alice can exercise the option successfully, or if the exercise fails, she is compensated with the guarantee deposited by Bob.
\end{restatable}

The proofs are included in \Cref{sec:option_proof}.

\section{EVALUATION}
\label{sec:impl}

\mypara{Implementation} 
We implemented, tested, and evaluated our proposed protocol\footnote{Code is available in \url{https://github.com/sfofgalaxy/cross-chain-options}}, Efficient Cross-Chain Options without Upfront Holder Collateral in \Cref{sec:ftcop}.
We evaluated a cross-chain over two ERC-20 tokens on Ethereum Virtual Machine (EVM)-compatible chains.
The experiment was conducted within the latest EVM and the Solidity Compiler version 0.8.22.
The on-chain transaction cost (gas consumption) in this environment is identical to the current Ethereum network.
DAPS can be implemented by algorithms like Schnorr, DSA, and EdDSA, which are natively supported by most blockchains.
Implementing DAPS is straightforward, as signature algorithms inherent in blockchains can directly support it.
We employ the same signature algorithm as Ethereum, utilizing the secp256k1 curve and the Elliptic Curve Digital Signature Algorithm (ECDSA)~\cite{Johnson2001}. 
Since EVM does not support direct on-chain verification of a public-private key pair, we implement the proof of the private key $sk$ by signing a specific message with it.

\mypara{Expected Transfer Time Evaluation} 
We compared our work with TCO~\cite{transferable}, since currently, only their protocol achieves transferability without the need for a trusted third-party cross-chain bridge.
Assuming that the probability of an option being transferred and ultimately finalized within the current network is $p$, the total number of transfers is $X$.
Then, in their protocol, $X$ follows a geometric distribution, i.e., $X\sim G(p)$.
The expected successful transfer time and the transfer probability of each phase are illustrated in \Cref{fig:prob1} and \Cref{fig:prob2}.
When many malicious nodes exist in the current network, say, the finalization probability is 10\%, the duration of the mutate lock phase and the consistency phase in their protocol becomes significantly prolonged, reaching $45\Delta$, which is approximately equivalent to 2 days in Bitcoin.
By initiating the replace phase earlier and consolidating the mutate and consistency phases, we significantly reduce the duration of these phases. 

\mypara{Gas Consumption Evaluation}
\Cref{fig:gas} lists the gas consumption for contract deployment, option operations, and gas consumption in different phases\footnote{In the transfer failure case, we only consider the gas consumption of conforming parties.}, where gas price is $1$ Gwei (average price)~\cite{etherscan} and ETH price is $\$2313.17$ on 24 Feb, 2025.

Notably, as shown in \Cref{fig:gas1}, compared to TCO, the gas consumption for the holder transfer process has significantly decreased from 714,867 to 510,857 gas (a reduction of approximately 28.5\% for successful transfers).
For failed transfers, the gas consumption decreases from 330,350 gas to 248,388 gas (a reduction of approximately 3.4\%).
The gas consumption of the transferring writer also decreases to a similar level.
The gas consumption for exercising an option increases from 96,916 to 145,337, while the gas consumption for abandonment decreases.
This is because, during an exercise, Alice needs to deposit funds, and Bob must fulfill the request by revealing the exercise secret.

In contrast, for an abandonment, Bob only needs to perform a refund operation.
\Cref{fig:gas2} illustrates the gas consumption of a successful transfer across different phases. In our protocol, the reveal phase only requires the seller to reveal a signature in one contract, significantly reducing gas consumption compared to the mutate and replace/revert phases, lowering the gas for the holder and the writer to 123,158 and 123,435 gas, respectively.
However, gas usage in the consistency phase is higher than that of their protocol, as we verify signatures in both contracts to ensure consistency.
The gas consumption for contract deployment in our protocol is generally higher compared to TCO due to additional security measures and DAPS support.
The deployment costs for $Contract_A$ and $Contract_B$ in our protocol are about 2.5M and 2.2M gas.

Nonetheless, the gas costs are acceptable, although some operations increase gas usage.
Our protocol implements more secure operations and reduces the costs during transfers and abandonment.
On the other hand, gas fees for each operation (usually less than \$6) are insignificant compared to the option's value.
Although gas costs for exercising options and deploying contracts are higher, this expense ensures efficient and robust transfers.
Additionally, it also eliminates the holder's collateral, which can be invested as other liquidity and increase the capital efficiency.
This addresses traders' primary concerns.

\begin{figure}[ht]
    \centering
    \subfloat[Holder Transfer]{
        \includegraphics[width=0.244\textwidth]{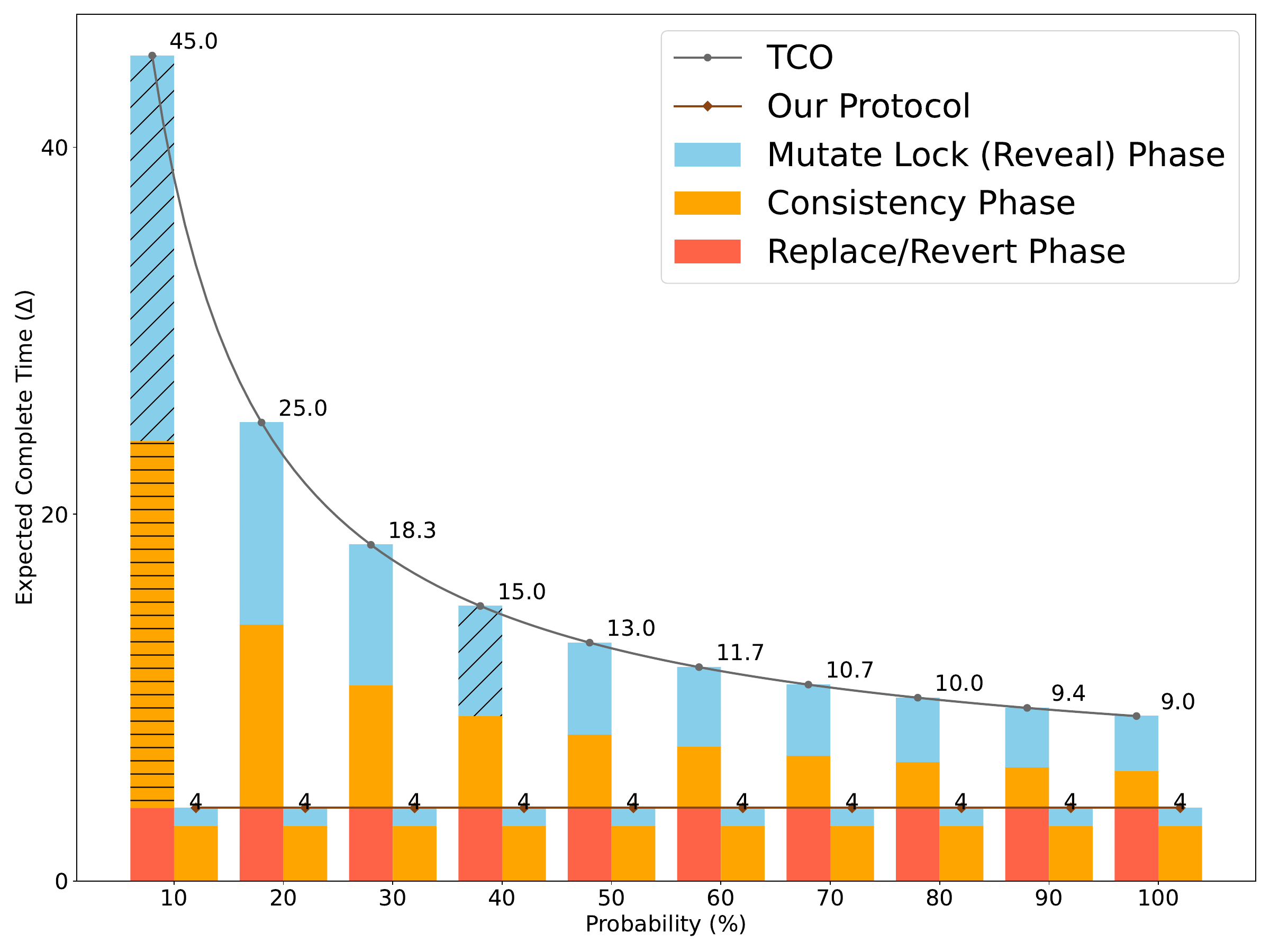}
        \label{fig:prob1}}
    \hspace*{-0.02\textwidth}
    \subfloat[Writer Transfer]{
        \includegraphics[width=0.244\textwidth]{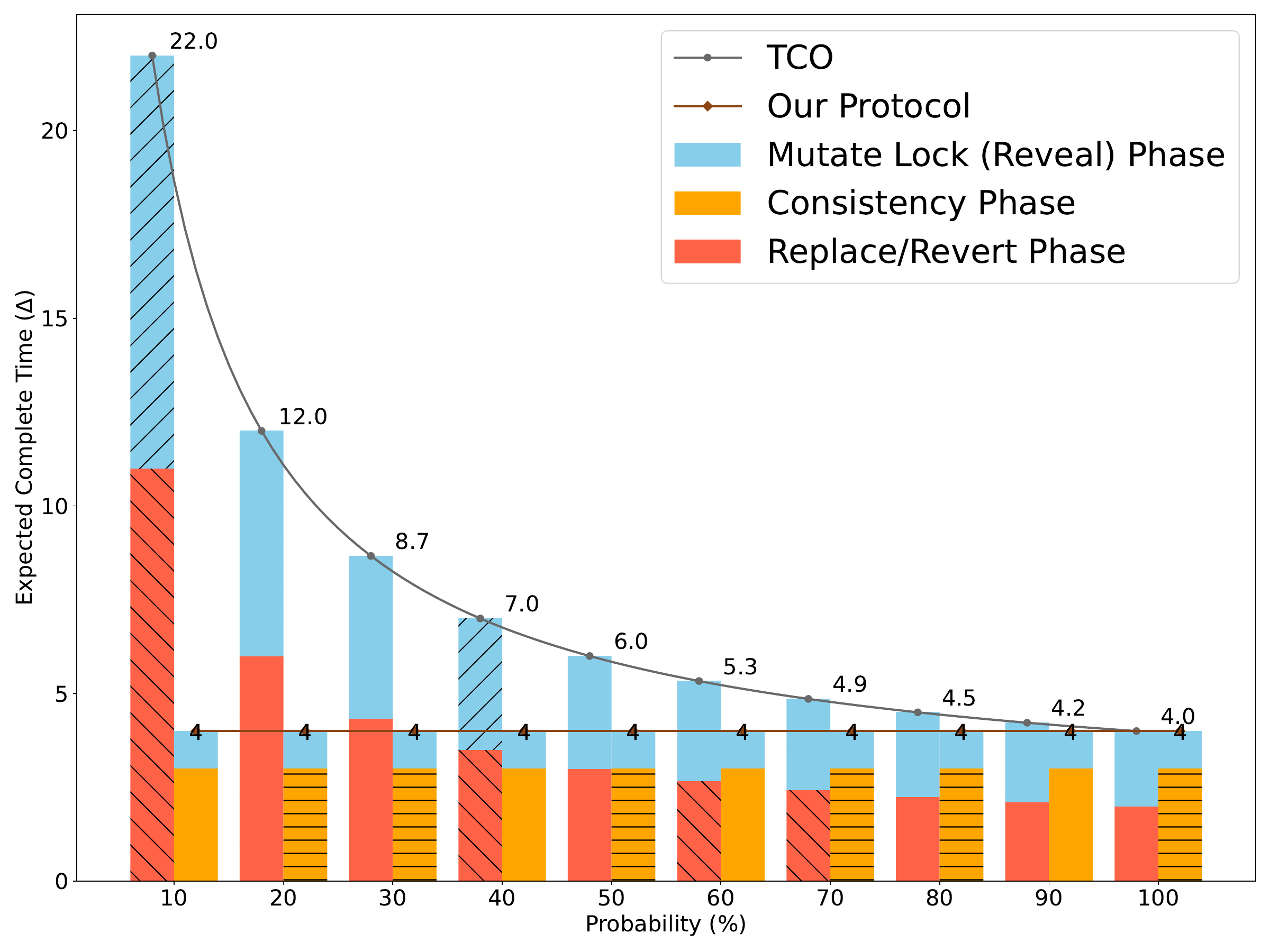}
        \label{fig:prob2}}
    \caption{Expected transfer time and successful probability.}
    \label{fig:prob}
\end{figure}

\begin{figure*}[ht]
    \centering    
    \subfloat[Transfer and Exercise/Abandon]{\includegraphics[width = .32\textwidth]{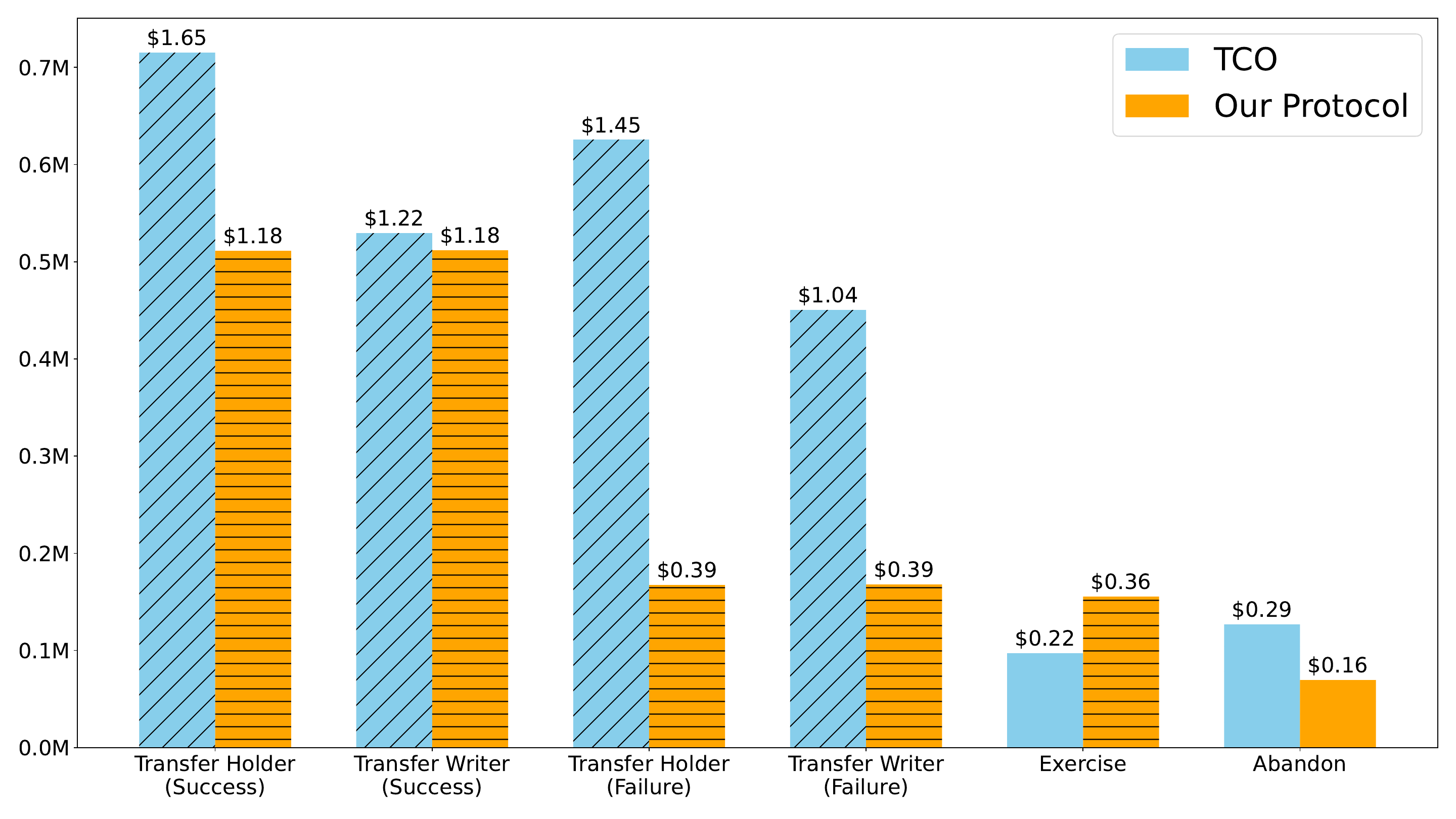}\label{fig:gas1}}
    \hfill
    \subfloat[In Different Phases]{ \includegraphics[width = .32\textwidth]{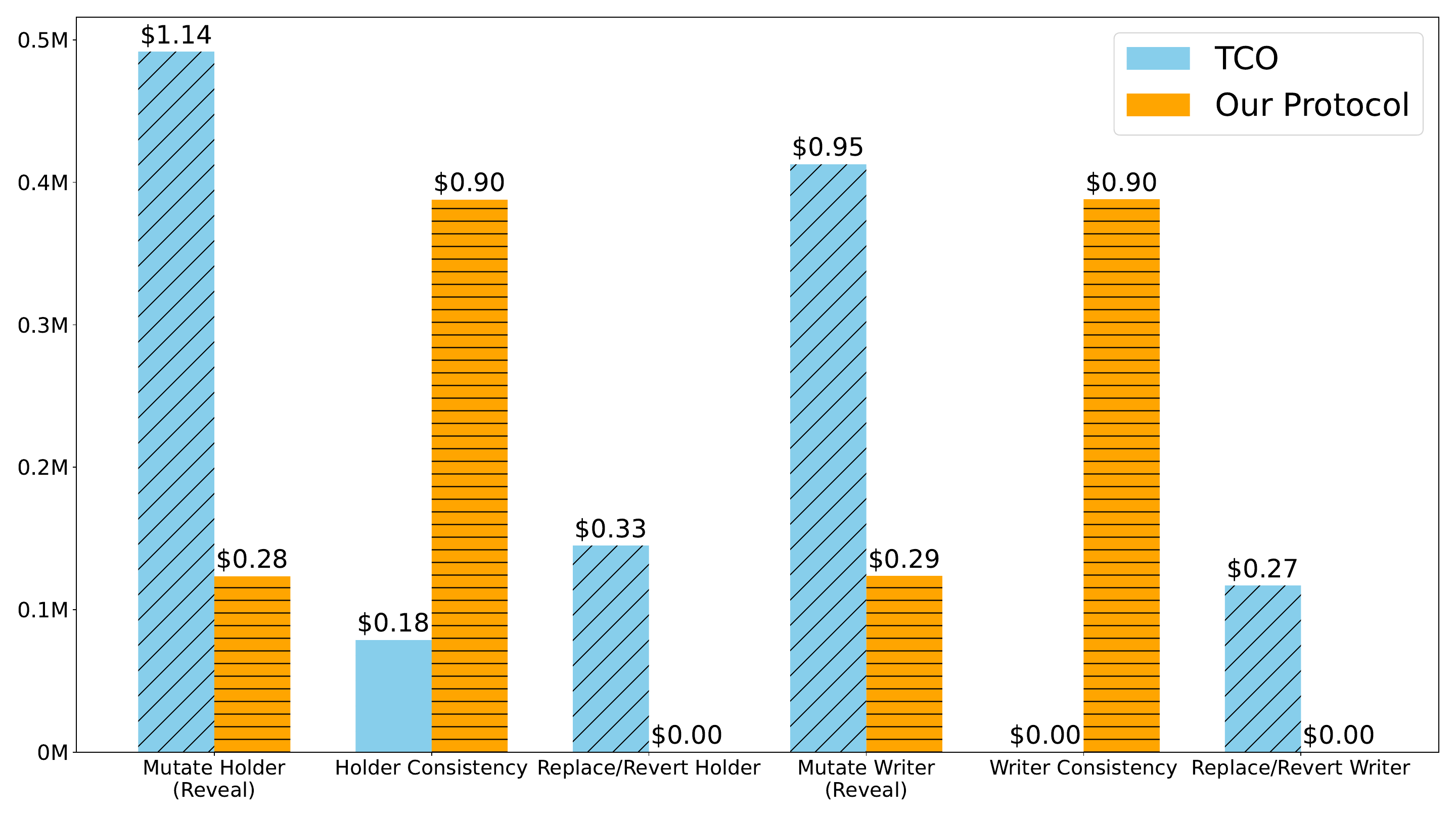}\label{fig:gas2}}
    \hfill
	\subfloat[Contract Deployment]{ \includegraphics[width = .32\textwidth]{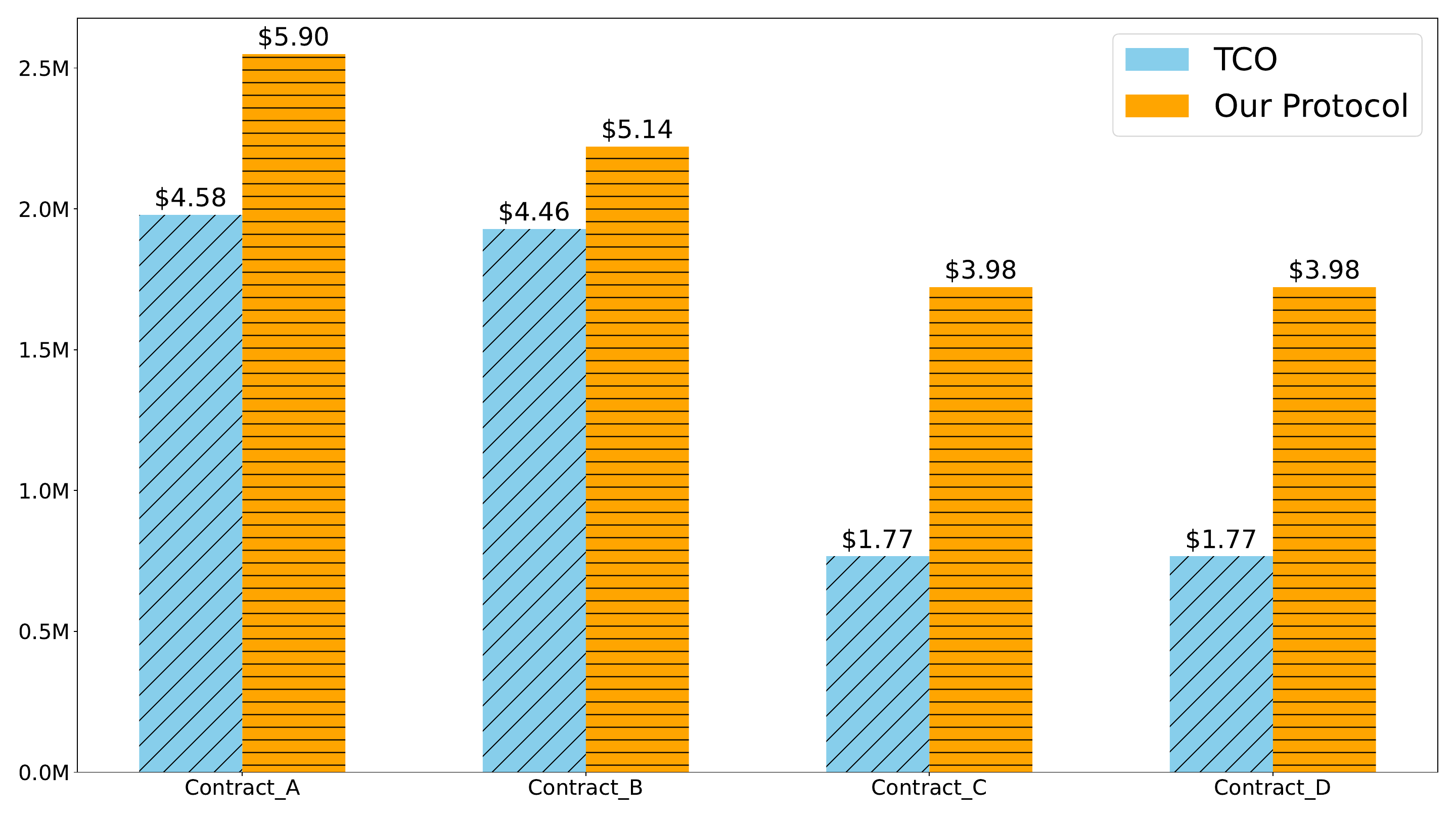}\label{fig:gas3}}
    \caption{Gas consumption with gas price of $1$ Gwei (on average) and ETH price of $\$2313.17$ (24 Feb, 2025)} 
    \label{fig:gas}
\end{figure*}

\section{RELATED WORK}
\label{sec:related}
Currently, many on-chain options trading platforms are available in the market.
They price the options using an Automatic Market Maker (AMM). Lyra~\cite{lyrafinance2024} stands as the preeminent decentralized options trading platform, commanding around a third of the market's TVL, and employs an AMM with a Black76~\cite{BLACK1976167} pricing model. However, its operations hinge on external data feeds from oracles, such as spot prices and implied volatility. Hegic~\cite{hegicwhitepaper} decentralizes the writers' risk and employs a fixed pricing rate based on option expiry date and target prices, which leads to less accurate pricing. 

In traditional markets, the price of options is determined by supply and demand.
Devising an effective pricing model for options by AMM faces challenges due to the lack of accurate supply and demand modeling.
Therefore, an order-book-based decentralized exchange shows up. Aevo~\cite{aevo} is a high-performance, order-book-based decentralized exchange, which closely resembles the traditional options market. However, its current implementation employs an off-chain orderbook coupled with on-chain settlement, which introduces a higher degree of centralized risk into the protocol.
Opyn~\cite{opyn} provides users with the ability to sell European options by minting \textit{ERC20} tokens as the option. These tokens can be destroyed to exercise rights or transacted in the market. However, the system faces challenges due to high gas fees on Ethereum and a lack of necessary liquidity for exchanges. All of the above on-chain protocols lack universality, most~\cite{lyrafinance2024, opyn, hegicwhitepaper, aevo} currently only support ETH and BTC options trading.

An HTLC can address the aforementioned issues by enabling two parties to create option contracts across two chains. These contracts lock assets, agreed upon by both parties, on two chains at a predetermined price.
HTLCs~\cite{fasterAtomic,Fault-Tolerant,bitcointtalk} were originally designed for cross-chain atomic swaps.
Subsequently, Han \textit{et al.}~\cite{FairnessAtomicSwaps} highlighted the optionality and fairness aspects for one party, demonstrating that an atomic cross-chain swap is equivalent to a premium-free American call option.
They estimate premiums with the Cox-Ross-Rubinstein option pricing model~\cite{COX1979229}.
They address the unfairness by incorporating a premium mechanism. In~\cite{SoreLoserAttacks}, the authors define a sore loser attack in cross-chain swaps and let participants escrow assets along with a negotiated option premium, which acts as compensation.
Nadahalli \textit{et al.}~\cite{griffree} separate the premium protocol from the collateral protocol, employing upfront communication of off-chain unspent transaction outputs as the option premium and collateral.
Other researchers~\cite{cfa, liu2020atomic} introduce cross-chain atomic options, incorporating concepts such as the holder's late margin deposit and early cancellation of the option.

In TCO~\cite{transferable}, the authors introduce the transferability of options. However, their approach requires long transfer times and does not support concurrent trading involving multiple buyers, which may lead to phantom bid attacks.
An adversary can create multiple fake buyers who offer higher prices but fail to complete the transfer. Consequently, the option holder is unable to sell their position.

None of these protocols eliminates the holder's collateral in cross-chain options.
To eliminate the holder's upfront collateral requirement, cross-chain transaction confirmation can be adopted to verify the collateral deposition on one chain when the option is exercised.
This approach can employ cross-chain bridges.
Some cross-chain bridges rely on external verification and introduce a trusted third party to facilitate message transmission.
This approach is vulnerable to many attacks~\cite{zhangbridges,SlowMist}, such as rug pulls~\cite{OrdiZK}, code vulnerabilities~\cite{polyHack}, and private key leakage~\cite{ALEX}.
Some bridges employ native verification and use light clients on both chains to verify proofs. This method requires complex smart contracts and incurs high verification and storage costs~\cite{daveasGasEfficientSuperlightBitcoin2020,Glimpse,zamyatinXCLAIMTrustlessInteroperable2019,zkBridge}.
The diversity and heterogeneity of blockchains significantly increase the time and cost of implementing a light client for each chain.
An alternative is using TEE for cross-chain transactions~\cite{tesseract,BDTF}.
Those solutions are susceptible to many vulnerabilities, including side-channel attacks, which present significant security risks~\cite{TEESecurity,TEECS,TEEsoftAttack,TEEAttackCrossTalk}.

Our proposed approach provides an efficient cross-chain option protocol by combining HTLC logic with a signature scheme. This combination facilitates the transfer of positions and replacement of hashlocks in option contracts. Our approach eliminates the need for option holders to provide upfront collateral. Instead of relying on cross-chain bridges, we achieve this through a distributed protocol design bolstered by economic incentives.

\section{CONCLUSION AND DISCUSSION}
\label{sec:conclusion}
In this paper, we propose an approach for efficient and collateral-free cross-chain options.
Compared to an HTLC-based option, our approach incorporates a signature scheme, DAPS, to facilitate efficient option transfer, reducing the transfer time to half of the previous work.
Akin to traditional options, our approach does not require the option holder to put a collateral upfront when an option is established, thus making the most of the advantages of options for leveraging and hedging.
With our proposed method, the option markets can be more efficient.

One key feature of our proposed protocol is its resilience against phantom bidder attacks. The protocol supports concurrent bidding and can finalize the option transfer in constant time, not influenced by potential malicious option buyers.

In practice, we can further utilize the liquidity of the writer's collateral.
For example, Bob, the option writer, can still use the funds in the contract for trading, anyone can flash loan Bob's collateral, completed in a single transaction.

The options we explored in this paper are covered options. In option markets, there are other types of options, for example, naked options, where sellers must maintain their margin accounts at a threshold according to market fluctuations, otherwise face liquidation. In the future, we plan to design naked options with liquidation features and consider how to ensure the rapid transfer of cross-chain options.

\bibliographystyle{IEEEtran}
\bibliography{reference}

\begin{thebibliography}{10}
\providecommand{\url}[1]{#1}
\csname url@samestyle\endcsname
\providecommand{\newblock}{\relax}
\providecommand{\bibinfo}[2]{#2}
\providecommand{\BIBentrySTDinterwordspacing}{\spaceskip=0pt\relax}
\providecommand{\BIBentryALTinterwordstretchfactor}{4}
\providecommand{\BIBentryALTinterwordspacing}{\spaceskip=\fontdimen2\font plus
\BIBentryALTinterwordstretchfactor\fontdimen3\font minus \fontdimen4\font\relax}
\providecommand{\BIBforeignlanguage}[2]{{%
\expandafter\ifx\csname l@#1\endcsname\relax
\typeout{** WARNING: IEEEtran.bst: No hyphenation pattern has been}%
\typeout{** loaded for the language `#1'. Using the pattern for}%
\typeout{** the default language instead.}%
\else
\language=\csname l@#1\endcsname
\fi
#2}}
\providecommand{\BIBdecl}{\relax}
\BIBdecl

\bibitem{bitcoinWhitepaper}
\BIBentryALTinterwordspacing
S.~Nakamoto. (2008) Bitcoin: A peer-to-peer electronic cash system. [Online]. Available: \url{https://bitcoin.org/bitcoin.pdf}
\BIBentrySTDinterwordspacing

\bibitem{ethereumWhitepaper}
\BIBentryALTinterwordspacing
V.~Buterin. (2014) Ethereum: A next-generation smart contract and decentralized application platform. [Online]. Available: \url{https://ethereum.org/en/whitepaper/}
\BIBentrySTDinterwordspacing

\bibitem{defillama}
\BIBentryALTinterwordspacing
DeFillama. (2025) Defillama website. [Online]. Available: \url{https://defillama.com/}
\BIBentrySTDinterwordspacing

\bibitem{lyrafinance2024}
\BIBentryALTinterwordspacing
S.~Dawson, D.~Romanowski, A.~Cheng, and V.~Abramov. (2023) Lyra v2. [Online]. Available: \url{https://lyra.finance/files/v2-whitepaper.pdf}
\BIBentrySTDinterwordspacing

\bibitem{opyn}
\BIBentryALTinterwordspacing
Opyn. (2025) Opyn. [Online]. Available: \url{https://www.opyn.co/}
\BIBentrySTDinterwordspacing

\bibitem{hegicwhitepaper}
\BIBentryALTinterwordspacing
M.~Wintermute. (2020) Hegic protocol whitepaper. [Online]. Available: \url{https://github.com/hegic/whitepaper/blob/master/Hegic%20Protocol%20Whitepaper.pdf}
\BIBentrySTDinterwordspacing

\bibitem{aevo}
\BIBentryALTinterwordspacing
Aevo. (2025) Aevo. [Online]. Available: \url{https://www.aevo.xyz/}
\BIBentrySTDinterwordspacing

\bibitem{deriprotocol_whitepaper_v4_2023}
\BIBentryALTinterwordspacing
0xAlpha, R.~Chen, and D.~Fang. (2023) Deri v4 whitepaper. [Online]. Available: \url{https://docs.deri.io/library/whitepaper}
\BIBentrySTDinterwordspacing

\bibitem{SlowMist}
\BIBentryALTinterwordspacing
SlowMist. (2025) Slowmist hacked. [Online]. Available: \url{https://hacked.slowmist.io/?c=Bridge}
\BIBentrySTDinterwordspacing

\bibitem{cross_chain_security}
N.~Li, M.~Qi, Z.~Xu, X.~Zhu, W.~Zhou, S.~Wen, and Y.~Xiang, ``Blockchain cross-chain bridge security: Challenges, solutions, and future outlook,'' \emph{Distrib. Ledger Technol.}, vol.~4, no.~1, Feb. 2025.

\bibitem{chainlinkcrosschain}
\BIBentryALTinterwordspacing
Chainlink. (2025) What is a cross-chain bridge? [Online]. Available: \url{https://chain.link/education-hub/cross-chain-bridge}
\BIBentrySTDinterwordspacing

\bibitem{zhangbridges}
M.~Zhang, X.~Zhang, Y.~Zhang, and Z.~Lin, ``Security of cross-chain bridges: Attack surfaces, defenses, and open problems,'' in \emph{Proceedings of the 27th International Symposium on Research in Attacks, Intrusions and Defenses}, 2024, p. 298–316.

\bibitem{OrdiZK}
\BIBentryALTinterwordspacing
O.~Knight. (2024) Bitcoin bridge ordizk suffers apparent \$1.4m rug pull, token crashes to zero: Certik. [Online]. Available: \url{https://x.com/CoinDesk/status/1765034858177253545}
\BIBentrySTDinterwordspacing

\bibitem{polyHack}
\BIBentryALTinterwordspacing
P.~Network. (2021) The root cause of poly network being hacked. [Online]. Available: \url{https://medium.com/poly-network/the-root-cause-of-poly-network-being-hacked-e30cf27468f0}
\BIBentrySTDinterwordspacing

\bibitem{ALEX}
\BIBentryALTinterwordspacing
N.~Mutual. (2024) Taking a closer look at alex lab exploit. [Online]. Available: \url{https://medium.com/neptune-mutual/taking-a-closer-look-at-alex-lab-exploit-5a7e4b7ea0ed}
\BIBentrySTDinterwordspacing

\bibitem{TEESecurity}
D.~Cerdeira, N.~Santos, P.~Fonseca, and S.~Pinto, ``Sok: Understanding the prevailing security vulnerabilities in trustzone-assisted tee systems,'' in \emph{Proceedings of the 41st IEEE Symposium on Security and Privacy}, 2020, pp. 1416--1432.

\bibitem{TEECS}
A.~Muñoz, R.~Ríos, R.~Román, and J.~López, ``A survey on the (in)security of trusted execution environments,'' \emph{Computers \& Security}, 2023.

\bibitem{TEEsoftAttack}
M.~Lipp, A.~Kogler, D.~Oswald, M.~Schwarz, C.~Easdon, C.~Canella, and D.~Gruss, ``Platypus: Software-based power side-channel attacks on x86,'' in \emph{Proceedings of the 42nd IEEE Symposium on Security and Privacy}, 2021, pp. 355--371.

\bibitem{daveasGasEfficientSuperlightBitcoin2020}
S.~Daveas, K.~Karantias, A.~Kiayias, and D.~Zindros, ``A {{Gas-Efficient Superlight Bitcoin Client}} in {{Solidity}},'' in \emph{Proceedings of the 2nd {{ACM Conference}} on {{Advances}} in {{Financial Technologies}}}, 2020.

\bibitem{Glimpse}
G.~Scaffino, L.~Aumayr, Z.~Avarikioti, and M.~Maffei, ``Glimpse: {On-Demand} {PoW} light client with {Constant-Size} storage for {DeFi},'' in \emph{Proceedings of the 32nd USENIX Security Symposium}, 2023.

\bibitem{zamyatinXCLAIMTrustlessInteroperable2019}
A.~Zamyatin, D.~Harz, J.~Lind, P.~Panayiotou, A.~Gervais, and W.~Knottenbelt, ``{{XCLAIM}}: {{Trustless}}, {{Interoperable}}, {{Cryptocurrency-Backed Assets}},'' in \emph{2019 {{IEEE Symposium}} on {{Security}} and {{Privacy}}}, 2019.

\bibitem{zkBridge}
T.~Xie, J.~Zhang, Z.~Cheng, F.~Zhang, Y.~Zhang, Y.~Jia, D.~Boneh, and D.~Song, ``zkbridge: Trustless cross-chain bridges made practical,'' in \emph{Proceedings of the 2022 ACM SIGSAC Conference on Computer and Communications Security}, 2022, p. 3003–3017.

\bibitem{AtomicSwaps}
M.~Herlihy, ``Atomic cross-chain swaps,'' in \emph{Proceedings of the 2018 ACM Symposium on Principles of Distributed Computing}, 2018.

\bibitem{Fault-Tolerant}
Y.~Xue, D.~Jin, and M.~Herlihy, ``Invited paper: Fault-tolerant and expressive cross-chain swaps,'' in \emph{Proceedings of the 24th International Conference on Distributed Computing and Networking}, 2023, p. 28–37.

\bibitem{nieto2018trustminimized}
\BIBentryALTinterwordspacing
Fernando. (2019) Trust-minimized derivatives. [Online]. Available: \url{https://gist.github.com/fernandonm/}
\BIBentrySTDinterwordspacing

\bibitem{FairnessAtomicSwaps}
R.~Han, H.~Lin, and J.~Yu, ``On the optionality and fairness of atomic swaps,'' in \emph{Proceedings of the 1st ACM Conference on Advances in Financial Technologies}, 2019, p. 62–75.

\bibitem{arwenwhitepaper}
E.~Heilman, S.~Lipmann, and S.~Goldberg, ``The arwen trading protocols,'' in \emph{Proceedings of the 24th Financial Cryptography and Data Security: 24th International Conference}, 2020, p. 156–173.

\bibitem{liu2020atomic}
J.~A. Liu, ``Atomic swaptions: Cryptocurrency derivatives,'' \emph{CoRR}, vol. abs/1807.08644, 2018.

\bibitem{cfa}
\BIBentryALTinterwordspacing
M.~Tefagh, F.~Bagheriesfandabadi, A.~Khajehpour, and M.~Abdi, ``Capital-free futures arbitrage,'' 10 2020. [Online]. Available: \url{https://www.researchgate.net/publication/344886866_Capital-free_Futures_Arbitrage}
\BIBentrySTDinterwordspacing

\bibitem{transferable}
D.~Engel and Y.~Xue, ``Transferable cross-chain options,'' in \emph{Proceedings of the 4th ACM Conference on Advances in Financial Technologies}, 2023.

\bibitem{DAPS}
B.~Poettering and D.~Stebila, ``Double-authentication-preventing signatures,'' in \emph{Computer Security - ESORICS 2014}, 2022, p. 436–453.

\bibitem{sokHacks}
S.-S. Lee, A.~Murashkin, M.~Derka, and J.~Gorzny, ``Sok: Not quite water under the bridge: Review of cross-chain bridge hacks,'' in \emph{2023 IEEE International Conference on Blockchain and Cryptocurrency}, 2023.

\bibitem{ronin}
\BIBentryALTinterwordspacing
C.~Roark. (2024) Ronin bridge hack caused by error in upgrade deployment script — verichains. [Online]. Available: \url{https://cointelegraph.com/news/ronin-bridge-hack-upgrade-script-verichains}
\BIBentrySTDinterwordspacing

\bibitem{5RugPulls}
\BIBentryALTinterwordspacing
G.~Roy. (2024) 5 “worst” rug-pulls in crypto. [Online]. Available: \url{https://www.securities.io/5-worst-rug-pulls-in-crypto/}
\BIBentrySTDinterwordspacing

\bibitem{Johnson2001}
D.~Johnson, A.~Menezes, and S.~Vanstone, ``The elliptic curve digital signature algorithm (ecdsa),'' \emph{International Journal of Information Security}, vol.~1, no.~1, pp. 36--63, Aug 2001.

\bibitem{etherscan}
\BIBentryALTinterwordspacing
Etherscan. (2025) Ethereum gas tracker. [Online]. Available: \url{https://etherscan.io/gastracker#chart_gasprice}
\BIBentrySTDinterwordspacing

\bibitem{BLACK1976167}
F.~Black, ``The pricing of commodity contracts,'' \emph{Journal of Financial Economics}, vol.~3, no.~1, pp. 167--179, 1976.

\bibitem{fasterAtomic}
S.~Mazumdar, ``Towards faster settlement in htlc-based cross-chain atomic swaps,'' in \emph{2022 IEEE 4th International Conference on Trust, Privacy and Security in Intelligent Systems, and Applications}, 2022.

\bibitem{bitcointtalk}
\BIBentryALTinterwordspacing
T.~Nolan. (2013) Alt chains and atomic transfers. [Online]. Available: \url{https://bitcointalk.org/index.php?%20topic=193281.0}
\BIBentrySTDinterwordspacing

\bibitem{COX1979229}
J.~C. Cox, S.~A. Ross, and M.~Rubinstein, ``Option pricing: A simplified approach,'' \emph{Journal of Financial Economics}, 1979.

\bibitem{SoreLoserAttacks}
Y.~Xue and M.~Herlihy, ``Hedging against sore loser attacks in cross-chain transactions,'' in \emph{Proceedings of the 2021 ACM Symposium on Principles of Distributed Computing}, 2021, p. 155–164.

\bibitem{griffree}
T.~Nadahalli, M.~Khabbazian, and R.~Wattenhofer, ``Grief-free atomic swaps,'' in \emph{Proceedings of the 6th International Conference on Blockchain and Cryptocurrency}, 2022, pp. 1--9.

\bibitem{tesseract}
I.~Bentov, Y.~Ji, F.~Zhang, L.~Breidenbach, P.~Daian, and A.~Juels, ``Tesseract: Real-time cryptocurrency exchange using trusted hardware,'' in \emph{Proceedings of the 2019 ACM SIGSAC Conference on Computer and Communications Security}, 2019, p. 1521–1538.

\bibitem{BDTF}
G.~Su, W.~Yang, Z.~Luo, Y.~Zhang, Z.~Bai, and Y.~Zhu, ``Bdtf: A blockchain-based data trading framework with trusted execution environment,'' in \emph{2020 16th International Conference on Mobility, Sensing and Networking}, 2020, pp. 92--97.

\bibitem{TEEAttackCrossTalk}
H.~Ragab, A.~Milburn, K.~Razavi, H.~Bos, and C.~Giuffrida, ``Crosstalk: Speculative data leaks across cores are real,'' in \emph{Proceedings of the 42nd IEEE Symposium on Security and Privacy}, 2021, pp. 1852--1867.

\end{thebibliography}

\appendix

\subsection{PROOFS}
\label{sec:proofs}

\subsubsection{Transfer Protocol Proofs}
\label{sec:trans_proofs}

\safety*
\begin{proof}
    \label{sec:safety}
    In the $Contract_A$, the following elements are defined:
    \begin{itemize}
        \item $T_E$: The expiration time of this option.
        \item \textit{exercise\_hashlock:} The hash lock of this option, which is the hash of a secret value known only to the writer.
        \item \textit{old\_exercise\_hashlock:} The hash lock of this option, which is the hash of a secret value known only to the writer.
        \item \textit{holder:} The holder can call $exercise()$ to exercise the option before $T_E$.
        \item \textit{guarantee:} The writer's asset, i.e. $Asset_G$, which can be any asset mutually agreed upon by the holder and writer as guarantee. This can include tokens, NFTs, or any other type of asset.
        \item \textit{writer:} The writer can use the secret value to call $refund()$ to retrieve the guarantee or retrieve it directly after $T_E+2\Delta$.
        \item \textit{collateral:} The collateral that Alice must deposit if she decides to exercise the option to purchase Bob's asset.
        \item \textit{holder\_transfer\_public\_key:} the transfer key of Alice, $pk_A$, used for verify the transfer signature of Alice to Carol.
        \item \textit{writer\_transfer\_public\_key:} New transfer key of Dave, $pk_D$, used for verify the transfer signature of Dave to others.
        \item \textit{old\_writer\_transfer\_public\_key:} Old transfer key of Bob, $pk_B$, used for verify the transfer signature of Bob to Dave, Within the period of one $\Delta$, during which the transfer signature must be submitted to this contract, we still need to record the old transfer public key in case of Bob's misbehavior.
        \item \textit{writer\_transfer\_time:} The writer transfer time, used for Alice to claim assets if there exits misbehavior of Bob.
    \end{itemize}
    In the $Contract_B$, there are other additional items:
    \begin{itemize}
        \item \textit{collateral:} The writer's collateral, i.e. $Asset_B$, it can be claimed by holder with preimage of hashlock.
        \item \textit{holder:} The holder can call $exercise()$ to exercise the option before $T_E$.
        \item \textit{writer:} The writer can call $refund()$ to retrieve the guarantee or retrieve it directly after $T_E+2\Delta$.
    \end{itemize}
    In the $Contract_D$, the following elements are defined:
    \begin{itemize}
        \item \textit{T\_W:} The deadline for seller to reveal signature.
        \item \textit{buyer:} writer position buyer, i.e. Dave.
        \item \textit{seller:} writer position seller, i.e. Bob.
        \item \textit{old\_exercise\_hashlock:} The hashlock of exercise, if Bob reveals during the transfer, Dave is able to reclaim with preimage.
        \item \textit{exercise\_hashlock:} The new hashlock of exercise, generated by Dave.
        \item \textit{old\_writer\_transfer\_public\_key:} Bob's transfer public key, used for verify the signature of Bob.
        \item \textit{writer\_transfer\_public\_key:} New transfer public key generated by Dave, used for replacing Bob's key.
        \item \textit{transfer\_time:} Used for record the time of transfer (the time reveal signature) and calculate the withdrawal delayed period.
    \end{itemize}
    Take Bob transferring his position to Dave as an example, since Bob deposit $Asset_G$ and $Asset_B$ into the contracts, which is more complex. Transferring Alice's position to Carol is more simple.
    
    By \Cref{lemma:atomicity}, if compliant Bob loses his position, he will at least obtain Dave's collateral during the writer transfer process. 
    
    If Dave is conforming, then if Bob acts maliciously on his own, Bob provides two different signatures to different buyers, Dave can reclaim the transfer fee with extracted $sk_B$ since $D$ records \textit{old\_writer\_transfer\_public\_key} i.e. $pk_B$. If Bob reveals $B$ at the same time during transfer process, then Dave can use $B$ to reclaim $\texttt{WriterFee}$ since $Contract_C$ records \textit{old\_exercise\_hashlock} i.e. $H(B)$. If Alice and Bob collude, they can use $sk_B$ or $B$ to withdraw $Asset_G$ and $Asset_B$. Then, Dave can observe $sk_B$ or $B$ and withdraw $\texttt{WriterFee}$ during withdrawal delay period since $Contract_D$ records \textit{transfer\_time}.

    If Alice is conforming, then If Bob provides two different signatures to different buyers, Alice can extract $sk_B$ and submit it to obtain $Asset_G$ and $Asset_B$. If Bob or Dave publishes one signature exclusively on either $Contract_A$ or $Contract_B$, Alice can forward this signature to another chain to make sure the exercise secret hashlocks are consistent on two chains. If Bob and Dave collude, they use two signatures to change the hashlock. During the withdrawal delay period, Alice can obtain $Asset_G$ and $Asset_B$ using the extracted $sk_B$.

    Transferring Alice's position to Carol is simpler, as Alice does not deposit assets into the option contracts and cannot modify the exercise secret hashlock. Carol only needs to ensure consistency between the holders on the two chains. Otherwise, she can extract $sk_A$ and refund the $\texttt{HolderFee}$ during the withdrawal delay period.
\end{proof}

\unobstructibility*
\begin{proof}
By \Cref{lemma:parti}, Bob's participation is not required, it is evident that Bob cannot block the process of transferring a holder's position.
During Bob's transfer to Dave, Alice can only obtain Bob's collateral by two different messages signed with $sk_B$ or the exercise secret.
If Bob is honest, he will neither leak $s_B$, sign multiple messages nor leak exercise secret.
Consequently, Alice cannot interrupt the transfer process.
\end{proof}

\sbindependence*
\begin{proof}
After $Alice_i$ transfers to $Alice_{i+1}$, the holder in the current option's $Contract_A$ and $Contract_B$ is updated to $\text{Alice}_{i+1}$, and the transfer key is known only to $Alice_{i+1}$. Therefore, after a holder transfer, $Alice_{i+1}$ can transfer the position to $Alice_{i+2}$ by re-performing Protocol \ref{sec:ftcop} with the transfer key of $Alice_{i+2}$. Similarly, after $Bob_j$ transfers to $Bob_{j+1}$ (holder Alice does not contest within $\Delta$), the writer in the current option's $Contract_A$ and $Contract_B$ is updated to $\text{Bob}_{j+1}$. At this point, only $sk_B^{j+1}$ or its signatures can be used for the next transfer. $Bob_{j+1}$ can also transfer the position by re-performing Protocol \ref{sec:ftcop} with the new transfer key.
\end{proof}

\begin{lemma}
\label{lemma:atomicity}
Protocols \ref{sec:ftcop} satisfy atomicity: If conforming Alice/Bob loses their position, she/he will be able to obtain Carol/Dave's collateral.
\end{lemma}

\begin{proof}
Following \Cref{theorem:correctness}, in transferring the holder position, after Carol correctly escrows the collateral, Alice temporarily locks the holder position in both contracts using $H(C)$. If Carol uses $C$ to obtain the position before $T_H$, then Alice will obtain Carol's collateral at $T_H + \Delta$. If Carol does not reveal $C$ before $T_H$, Alice will not receive Carol's collateral. Similarly, in transferring the writer position, if Bob does not reveal his signature honestly, then Bob will lose the position and Dave can retrieve and will not lose the collateral. If honest Bob signs for a buyer Dave, the honest Dave will use the signature to obtain Bob's position at $T_W$. Bob will then obtain Dave's collateral at $T_W + \Delta$.
\end{proof}

\counterparty*
\begin{proof}
Suppose both Carol and Dave are interested in Alice's and Bob's positions, respectively. 
According to \Cref{lemma:parti}, Alice transferring to Carol does not require Bob's involvement, hence Alice and Carol will not be interfered with. 
Similarly, it is known that during Bob's transfer to Dave, by \Cref{lemma:alice_parti}, if Bob and Dave are compliant, Alice does not need to participate. 
Considering the case when Bob reveals two different signatures: (i) If Carol has already revealed the secret value $C$ of the transfer hash lock, then Carol becomes the new holder and can use two different signatures by $s_B$ to obtain $Asset_B$ and $Asset_G$. 
(ii) If Carol has not revealed $C$ and will reveal it after $\Delta$, Carol can simultaneously reveal $C$ and call $reclaim()$ on both chains after $\Delta$ to obtain $Asset_B$ and $Asset_G$. 
If Dave or Bob publishes $\sigma_m$ on one single chain, Carol must forward $\sigma_m$ to the other chain while revealing $C$.
\end{proof}

\subsubsection{Option Properties Proofs}
\label{sec:option_proof}

\opcor*
\begin{proof}
According to Protocol \ref{sec:cfop}, it is evident that if Alice escrows her collateral in the $Contract_A$ contract and calls $exercise()$, then a conforming Bob will reveal the pre-image $B$ in $Contract_A$ to reclaim the guarantee $Asset_G$ and Alice's collateral $Asset_A$. Subsequently, Alice can use $B$ to obtain $Asset_B$. If Alice does not escrow the collateral, Bob will not reveal $B$. After the option expires at $T_E + 2\Delta$, Bob can call $claim()$ and $refund()$ on the respective chains to reclaim $Asset_G$ and $Asset_B$.
\end{proof}

\liveness*
\begin{proof}
According to Protocol \ref{sec:ftcop}, during the transfer from Bob to Dave, Alice can make a deposit and exercise her option at any time. If the transfer is in the Setup Phase, Bob will need to reveal $B$ to fulfill his obligation and revoke the transfer. It is important to note that Dave can use $B$ to reclaim $Trans_W$. If the transfer is in the Attempt Phase and Bob acts maliciously by using $B$ to take $Asset_G$, Alice can use $B$ to obtain $Asset_B$. Dave will need to use $B$ on $Contract_B$ to withdraw the transfer. Otherwise, when Dave uses $\sigma_m$ to change the writer and the hash lock, he will reveal a new preimage secret $D$, which Alice can then use to obtain $Asset_B$.
\end{proof}

\compensation*
\begin{proof}
By \Cref{theorem:opcor}, if Alice successfully exercises her option, she will receive Bob's collateral. Otherwise, after Alice makes a deposit and calls $exercise()$, $Contract_A$ can invoke $isDeposited()$ to determine if the exercise has occurred. If Bob does not fulfill his obligation within a period of $\Delta$, Alice can call $claim()$ to obtain $Asset_G$ as compensation, and Bob will lose his guarantee.
\end{proof}

\end{document}